\documentclass[runningheads]{jgaa-art-plane}
\usepackage[T1]{fontenc}
\usepackage{amsmath,amssymb}
\usepackage{subcaption}
\usepackage{xspace}
\PassOptionsToPackage{dvipsnames}{xcolor}
\DeclareGraphicsExtensions{.pdf,.png}
\graphicspath{{figs/}{notes/}}
\usepackage{placeins}
\usepackage{subcaption}
\usepackage{enumerate}

\definecolor{dark blue}{rgb}{0.121,0.47,0.705} 
\newcommand{\bl}{\color{dark blue}}
\definecolor{dark red}{rgb}{0.89,0.102,0.109}
\newcommand{\re}{\color{dark red}}
\definecolor{dark orange}{rgb}{1,0.498,0}
\newcommand{\og}{\color{dark orange}}
\definecolor{dark green}{rgb}{0.2,0.627,0.172}
\newcommand{\gr}{\color{dark green}}
\definecolor{dark cyan}{rgb}{0.106,0.62,0.467}
\newcommand{\cy}{\color{dark cyan}}
\definecolor{dark brown}{rgb}{0.651,0.337,0.157}
\newcommand{\br}{\color{dark brown}}
\definecolor{dark purple}{rgb}{0.415,0.239,0.603}
\newcommand{\pu}{\color{dark purple}}
\definecolor{defcolor}{HTML}{54278f}
\usepackage{mathtools}
\DeclarePairedDelimiter\set{\lbrace}{\rbrace}
\DeclarePairedDelimiter\abs{\lvert}{\rvert}

\DeclarePairedDelimiter\croc{\langle}{\rangle}
\def\Oh{\ensuremath{\mathcal{O}}}
\def\bT{\ensuremath{\mathbb{T}}\xspace}
\def\bQ{\ensuremath{\mathbb{Q}}\xspace}
\def\bR{\ensuremath{\mathbb{R}}\xspace}
\def\bZ{\ensuremath{\mathbb{Z}}\xspace}
\def\cC{\ensuremath{\mathcal{C}}\xspace}
\def\cD{\ensuremath{\mathcal{D}}\xspace}

\def\cT{\ensuremath{\mathcal{T}}\xspace}
\def\cL{\ensuremath{\mathcal{L}}\xspace}

\def\R{\ensuremath{\mathcal{R}}\xspace}

\def\blG{\ensuremath{{\bl L_1(G)}}\xspace}
\def\reG{\ensuremath{{\re L_2(G)}}\xspace}
\def\para{\!\parallel\!}
\def\hati{\hat{\imath}} % or \hat{i}
\def\upward{orbit-advancing\xspace}
\def\enclosing{orbit-enclosing\xspace}

\newtheorem{observation}[theorem]{Observation}

\newtheorem{proposition}[theorem]{Proposition}

\usepackage{thmtools} 
\usepackage{url}
\usepackage{cite}
\newcommand{\etal}{{et~al.}\xspace}

\PassOptionsToPackage{bookmarksopen=true,
			bookmarksopenlevel=2,}{hyperref}
\usepackage[capitalise,nameinlink]{cleveref}
\crefname{numclaim}{Claim}{Claims}
\Crefname{numclaim}{Claim}{Claims}
\crefname{observation}{Observation}{Observations}
\Crefname{observation}{Observation}{Observations}
\crefname{proposition}{Proposition}{Propositions}
\Crefname{proposition}{Proposition}{Propositions}
\crefname{corollary}{Corollary}{Corollaries}
\Crefname{corollary}{Corollary}{Corollaries}

\begin{document}

\doi{xxx}
\Issue{0}{0}{0}{0}{0} 
\title{Rectangular Duals on the Cylinder and the Torus}
\HeadingTitle{Rectangular Duals on the Cylinder and the Torus}
\HeadingAuthor{T. Biedl, P. Kindermann, and J. Klawitter}
\authorOrcid[first]{Therese Biedl}{}{0000-0002-9003-3783}
\affiliation[first]{University of Waterloo, Canada}
\authorOrcid[second]{Philipp Kindermann}{}{0000-0002-9003-3783}
\affiliation[second]{Universität Trier, Germany}
\authorOrcid[third]{Jonathan~Klawitter}{}{0000-0001-8917-5269} 
\affiliation[third]{University of Auckland, Aotearoa New Zealand}
\Ack{Jonathan Klawitter was supported by the 
Beyond Prediction Data Science Research Programme (MBIE grant UOAX1932).
Therese Biedl is supported by NSERC, FRN RGPIN-2020-03958.
}
\maketitle

\pdfbookmark[1]{Abstract}{Abstract}
\begin{abstract}
A rectangular dual of a plane graph $G$ is a contact representation of~$G$ by interior-disjoint rectangles
such that (i)~no four rectangles share a point, and (ii)~the union of all rectangles is a rectangle.
In this paper, we study rectangular duals of graphs that are embedded in surfaces other than the plane.
In particular, we fully characterize when a graph embedded on a cylinder admits a cylindrical rectangular dual.
For graphs embedded on the flat torus, we can test whether the graph has a toroidal rectangular dual 
if we are additionally given a \textit{regular edge labeling}, 
i.e.\ a combinatorial description of rectangle adjacencies.
Furthermore we can test whether there exists a toroidal rectangular dual that respects the embedding 
and that resides on a flat torus for which the sides are axis-aligned. 
Testing and constructing the rectangular dual, if applicable, can be done~efficiently. 
\end{abstract}

\begin{figure}[h]
  \centering
  \begin{subfigure}[t]{0.4 \linewidth}
	\centering
  	\includegraphics[page=1]{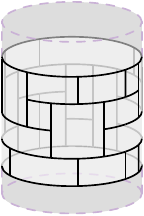}
  \end{subfigure}
  \begin{subfigure}[t]{0.4 \linewidth}
	\centering
  	\includegraphics[page=2]{cylinderTorus}
  \end{subfigure}
  \caption{A rectangular dual on the cylinder and on the torus.}
  \label{fig:teaser}
\end{figure}

\section{Introduction} % ------------------------
\label{sec:intro}
In a \emph{contact representation} $\R$ of a graph $G$ each vertex $w$ of~$G$ is mapped to an object $\R(w)$
such that two vertices $u$ and $v$ are adjacent in $G$ if and only if $\R(u)$ and $\R(v)$ touch,
that is, they intersect with disjoint interiors; all other objects are disjoint.
A famous example are Koebe's coin graphs~\cite{Koe36}, which can represent all planar graphs. 
Further examples include contact graphs of triangles~\cite{dFdMR94,SHRMH18}, squares and rectangles~\cite{BGPV08,KNU15,dLDEJ17},
line segments and arcs~\cite{Hli01,AE0KPSU15}, and unit disks~\cite{HK01,Epp17}.
A \emph{rectangular dual} of a plane graph $G$ is a commonly studied type of contact representation
where each vertex~$v$ is mapped to a non-empty, axis-aligned rectangle satisfying two conditions (see \cref{fig:recduals}):   
\begin{enumerate}[(i)]
  \item No four rectangles share a point;
  \item the union of the rectangles is a rectangle. 
\end{enumerate}
It has been characterized repeatedly when a plane graph~$G$ admits a rectangular dual~\cite{He93,LL90,KK85,Tho86,Ung53,KS24},
and how to test these conditions and construct the rectangular dual in linear time~\cite{BS87,KH97,BD16,KS24}.
Rectangular duals have originally been studied due to their relation to floor plans in architecture~\cite{Ste73}
and VLSI design~\cite{LL84,YS95}, and to rectangular cartograms from cartography~\cite{GS69,NK16}.
Here we are interested in generalizing rectangular duals from the plane to other surfaces,
specifically, to the cylinder and the torus as shown in~\cref{fig:teaser}. 

A common way to represent the surface of the cylinder [torus] is with a parallelogram~$Q$,
called the \emph{fundamental} or \emph{canonical polygon}, 
where one pair [both pairs] of opposite sides of~$Q$ are identified;
imagine cutting the cylinder [torus] along a path [two cycles] and then mapping the surface onto~$Q$.
We also speak of~$Q$ as the \emph{flat cylinder} or \emph{flat torus}, accordingly.
If~$Q$ has axis-parallel sides, we call it the \emph{rectangular flat cylinder/torus};
otherwise, we call it \emph{slanted}.
We assume that we have a standing cylinder and
thus call the identified sides of a flat cylinder~$Q$ the left and right side.

A \emph{cylindrical} [\emph{toroidal}$\,$] graph~$G$ is a graph embedded on the cylinder [torus] such that no two edges cross. 
Note that while a planar triangulated graph has a unique embedding in the plane (up to the choice of the outer face),
this is not the case for a toroidal triangulated graph:
Given one embedding, we can obtain another one by, roughly speaking,
cutting the torus along a cycle through the hole (or around the body),
twisting it and glueing it back together; this operation is known as a \emph{Dehn twist}.
Hence, when working with a toroidal graph we may also need a combinatorial description 
of how the graph is embedded on the flat torus (defined in \cref{sec:prelim}).
For rectangular duals of cylindrical and toroidal graphs,
we consider representations on a fundamental polygon~$Q$
that maintain the characteristic properties of planar rectangular duals.
In particular, we assume that no four rectangles share a point.
Similar to how a planar rectangular dual is bounded by four axis-aligned line segments,
we want both the top and bottom of a cylindrical rectangular dual to be bounded by a horizontal line.
On the torus on the other hand, we have no border and thus require the rectangular dual to fully cover~$Q$.
See \cref{fig:teaser,fig:recduals,fig:otherrecduals} for examples and~\cref{sec:prelim} for precise definitions.
Since rectangles may extend across an (identified) side of the fundamental polygon,
we can, unlike in the case of planar rectangular duals, realize
\emph{separating triangles} (i.e.\ a triangle that is not the boundary of a face), or also loops or parallel edges. 

\begin{figure}[t]
  \captionsetup[subfigure]{justification=centering}
  \centering
  \begin{subfigure}[t]{0.31 \linewidth}
	\centering
  	\includegraphics[page=1]{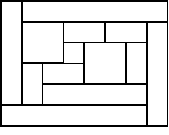}
 	\caption{Plane rectangular dual}
	\label{fig:recduals:plane}
  \end{subfigure}
  \hfill
  \begin{subfigure}[t]{0.31 \linewidth}
	\centering
    \includegraphics[page=2]{recDuals}
	\caption{Cylindrical \\ rectangular dual}
	\label{fig:recduals:cylinder}
  \end{subfigure}
  \hfill
  \begin{subfigure}[t]{0.31 \linewidth}
	\centering
  	\includegraphics[page=3]{recDuals}
	\caption{Toroidal \\ rectangular dual}
	\label{fig:recduals:torus}
  \end{subfigure}
  \caption{Rectangular duals on different surfaces.}
  \label{fig:recduals}
\end{figure}

An important notion when working with a rectangular dual is a \emph{regular edge labeling (REL)},
which, roughly speaking, is a 2-coloring and orientation of the edges that combinatorially
describe how rectangles touch (see \cref{sec:prelim} for details). 
Introduced by He~\cite{He93}, a REL can be found (if one exists) for planar graphs 
and used to construct rectangular duals in linear time~\cite{KH97};
any rectangular dual naturally gives rise to a REL.
In addition, RELs are fundamental tool in a variety of related applications, such as 
area-universality, constrained drawings, 
simultaneous and partial representation extensions, and morphing~\cite{KH97,Fus09,EMSV12,CFKKRW22,CKKRW22}.
Bernard, Fusy, and Liang~\cite{BFL24} generalized Schnyder woods and REL as well as other combinatorial structures of planar graphs to Grand Schnyder woods. 
Bonichon and L{\'{e}}v{\^{e}}que~\cite{BL19} studied toroidal RELs,
yet did not explore RELs as a tool to construct toroidal rectangular duals.
This and also the recent work on visibility representations on the torus and Klein bottle~\cite{Bie22}, % jk: why is there a comma?
motivated this~work.

\paragraph{Contribution.}
In this paper, we characterize exactly when a cylindrical graph (with fixed embedding)
has a rectangular dual on the flat cylinder. 
For toroidal graphs, a similar characterization remains open, 
but we provide a significant and non-trivial step by arguing when a given REL can be realized as a rectangular dual.   
More precisely, given a toroidal graph~$G$, a REL~$\cL$, and an embedding on a flat torus,
we can decide in linear time whether~$G$ has a rectangular dual~$\R$ on a flat torus that \emph{realizes}~$\cL$, 
i.e.\ the REL that arises from~$\R$ is~$\cL$.

While the flat torus in \cref{fig:otherrecduals} is \textit{slanted} (some of its sides are not axis-parallel),
it would be preferred to have a drawing on the \textit{rectangular flat torus}.
Yet there are toroidal graphs (such as the one in~\cref{fig:otherrecduals}) that have a toroidal rectangular dual, 
but none on the rectangular flat torus.
So we also study the question how to test whether a toroidal graph~$G$ 
has a rectangular dual on a rectangular flat torus.
We show that if given the REL~$\cL$ and an embedding on a rectangular flat torus,
then we can test whether we can find a rectangular dual that respects these input~constraints.

\begin{figure}[tbh]
  \captionsetup[subfigure]{justification=centering}
  \centering
  \begin{subfigure}[t]{0.34	 \linewidth}
	\centering
   	\includegraphics[page=1]{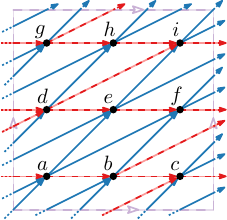}
  \end{subfigure}
  \hfill
  \begin{subfigure}[t]{0.63 \linewidth}
	\centering
   	\includegraphics[page=2]{noOrbitalRELexists}
  \end{subfigure}
  \caption{While this graph does not admit a rectangular dual on a rectangular flat torus,
  with the shown REL it admits one on the slanted flat torus. Red/blue edges are dashed/solid.}
  \label{fig:otherrecduals}
\end{figure}

In terms of presentation, we first study toroidal graphs in \cref{sec:torus}, 
since some results for cylindrical graphs (in \cref{sec:cylinder}) use them as a subroutine.   
Our proofs are constructive, and the corresponding representations can be found in linear or quadratic time.

\paragraph{Related Results.} 
For many graph drawing questions, it has been asked whether the known results for planar
graphs can be transferred to graphs embedded on other surfaces, for example 
for straight-line drawings~\cite{CEGL11,CDF13,DGK11,GL14,CDMB20}, level-planarity~\cite{ADDFPR20},
visibility representations~\cite{Bie22,Dean00,MR98,TT91}, and morphing~\cite{EL23}.
There is little previous work on contact representations and none (to our knowledge) concerning rectangular duals
on the torus or the cylinder~\cite{KP96,MR98}.
The closest related results are those on \emph{tessellation representations} of toroidal graphs~\cite{MR98}.   
Here we assign axis-aligned rectangles~$\R(e)$ to edges $e$ with some restrictions on how these touch; 
however, tessellation representations cannot be used to obtain rectangular duals or vice versa.

\section{Preliminaries} % ------------------------
\label{sec:prelim}
We start with a little primer on working on the flat torus, in particular, 
the classification of curves in homotopy classes and important basic results.
This is based on Stillwell~\cite{stillwell} and Chambers, Erickson, Lin, and Parsa~\cite{CELP21}.
We then give necessary definitions and notation on graphs, rectangular duals, 
and regular edge labelings, and make some basic observations on cycles in toroidal graphs on the flat~torus.

\subsection{Working with the Flat Torus} % - - - - - - - - - - - -
The (unit-square) flat torus~$\bT$ is the metric space obtained
by identifying opposites sides of the unit square~$[0,1]^2$ in the Eucledian plane via $(x, 0) \sim (x, 1)$ and $(0, y) \sim (1, y)$.
The \emph{universal covering plane (UCP)} is obtained by tiling the plane with infinitely many copies of~$\bT$.
So equivalently, $\bT$ is the quotient space~$\bT = \bR^2 / \bZ^2$.
The \emph{covering map} (or \emph{projection map}) is the function~$\pi \colon \bR^2 \to \bT$ via~$\pi(x, y) = (x \bmod 1, y \bmod 1)$.
We assume the unit-square flat torus~$\bT$ whenever we consider toroidal graph embeddings,
yet for some toroidal rectangular duals, we need a non-square flat torus~$Q$.
For those, UCP and covering map can be defined correspondingly by tiling the plane with copies of~$Q$, and defining the
covering map as the inverse of this operation.
If a toroidal graph is given on the surface of a torus in~$\bR^3$,
then one can compute an embedding on the flat torus in linear time~\cite{LPVV01}.

\pdfbookmark[3]{Curves}{Curves} 
\paragraph{Curves.} % -  -  -  -  -  -  -  -  -  - 
Throughout, we call a closed, simple, directed curve on a surface~$\Sigma$ (plane/cy\-lin\-der/torus) simply a \emph{curve}.
Two curves~$\cC, \cC'$ on~$\Sigma$ are \emph{homotopic} if one can be continuously deformed into the other.
More formally, interpreting a curve as a map from $[0,1]$ to~$\Sigma$, the curves~$\cC, \cC'$ are \emph{homotopic} 
if there exists a continuous map, a \emph{homotopy}, $f \colon [0,1] \times [0,1] \to \Sigma$ where $f(0,\cdot) = \cC$ and $f(1,\cdot) = \cC'$.
A curve is \emph{contractible} if it is homotopic to a single point and \emph{non-contractible} otherwise.

Any curve~$\cC$ on~$\bT$ can be \emph{lifted} to a curve in the UCP
by continuing on the adjacent copy of~$\bT$ rather than continuing on the other side.
We call this a \emph{covering curve} of~$\cC$ and use $\tilde{\cC}$ to denote it.
There are infinitely many covering curves of~$\cC$, depending on which copy of~$G$ and which point of~$\cC$ we use to start building it.
Sometimes we use the notation $\overline{\cC}$ for one copy of $\cC$ in the UCP, obtained by picking an arbitrary point $p$ on some
covering curve $\tilde{\cC}$, and then walking from point~$p$ along~$\tilde{\cC}$ until we have walked the entire length of $\cC$.
Note that a covering curve $\tilde{\cC}$ of a non-contractible curve $\cC$ goes towards infinity in both direction,
while a copy $\overline{\cC}$ only walks the length of $\cC$ once.

\pdfbookmark[3]{Crossings of Curves}{Crossings} 
\paragraph{Crossings.} % -  -  -  -  -  -  -  -  -  - 
For a curve $\cC$, let~$\cC[a,b)$ be the part of~$\cC$ between points~$a$ and~$b$, 
including the endpoint if we use a bracket and excluding it if we use a parenthesis.  
For two directed curves~$\cC$ and~$\cC'$ with a maximal common sub-curve~$\cC[p,q]$ (possibly~$p = q$),
we say that~$\cC$ \emph{crosses}~$\cC'$ \emph{at~$\cC[p,q]$} 
if~$\cC$ enters~$\cC'$ from one side at~$p$ and leaves it at~$q$ towards the opposite side;
otherwise the two curves are said to \emph{touch}.
We call such a crossing \emph{left-to-right} if~$\cC$ enters~$p$ left of~$\cC'$ and 
leaves~$q$ right of~$\cC'$; otherwise we call it \emph{right-to-left}. 
Summing over all maximal shared sub-curves, the \emph{algebraic crossing number} of~$\cC$ and~$\cC'$ is
\[ \hati(\cC, \cC')  = \abs{\set{\cC \text{ crosses } \cC' \text{ left-to-right} }} 
					- \abs{\set{\cC \text{ crosses } \cC' \text{ right-to-left} }}\text{.} \]
Correspondingly, we say that~$\cC$ \emph{crosses}~$\cC'$ \emph{algebraically left-to-right [right-to-left]}
if~$\hati(\cC, \cC') > 0$ [$\hati(\cC, \cC') < 0$].
Two homotopic curves $C, C'$ may cross each other, but not algebraically, 
if 
we have $\hati(\cC,\cC') = 0$.

\pdfbookmark[3]{Orbits and Homotopy Classes}{Orbits} 
\paragraph{Orbits and Homotopy Classes.} % -  -  -  -  -  -  -  -  -  - 
The sides of~$\bT$ are two non-contractible curves~$\bl M$ and~$\re H$ that intersect exactly once. 
The vertical curve~$\bl M$ is called \emph{meridian} or \emph{first orbit} and we assume~$\bl M$ is directed bottom-to-top.
The horizontal curve~$\re H$ is called \emph{horizon} or \emph{second orbit} (or sometimes \emph{latitude})
and we assume~$\re H$ is directed left-to-right.
(For a slanted flat torus, the naming is such that~$\re H$ is not vertical and~$\bl M$ is not horizontal.)

It is known that the pair ${\bl M}$ and ${\re H}$ are the generators of the fundamental group of $\bT$
which is isomorphic to $\mathbb{Z} \times \mathbb{Z}$ and coincide with its homology group.
We do not define these terms precisely, but only express its implications to the extent that we need them.
For a directed curve~$\cC$, define~$m(\cC) := \hati(\cC, {\bl M})$ and~$h(\cC) := \hati({\re H}, \cC)$,
that is, the number of times $\cC$ (algebraically) crosses the meridian left-to-right and the horizon right-to-left (bottom-to-top).
(Note that the order of curve and orbit are \textit{not} the same for both.)
The pair~$(m(\cC), h(\cC))$ is called the \emph{homotopy class} of~$\cC$ and always coprime.

The following observations are well known~\cite{stillwell}.
Orbits~$\bl M$ and~$\re H$ have homotopy class~$(0,1)$ and~$(1,0)$, respectively. 
Contractible curves have homotopy class $(0,0)$.
We call two curves~$\cC,\cC'$ \emph{parallel} if they have the same homotopy class
and \emph{reverse-parallel} if~$(m(\cC), h(\cC)) = {-}(m(\cC'), h(\cC'))$.
Two curves~$\cC$ and~$\cC'$ with~$\hati(\cC, \cC') = 0$ are either parallel, reverse-parallel,
or at least one is contractible.
For two [reverse-]parallel curves~$\cC, \cC'$ and third curve $\cD$,
it holds that $\hati(\cC, \cD) = \hati(\cC', \cD)$ [resp. $\hati(\cC, \cD) = -\hati(\cC', \cD)$].
Furthermore, the algebraic crossing number of two curves~$\cC, \cC'$
can be computed from the homotopy classes alone with 
$$ \hati(\cC, \cC') = m(\cC) h(\cC') - h(\cC)m(\cC')\text{.} $$

Consider a copy $\overline{\cC}$ of a curve~$\cC$ of homotopy class~$(m(\cC), h(\cC))$ in the UCP.
Note that it crosses the meridian algebraically left-to-right exactly $m(\cC)$ times
and the horizon algebraically right-to-left exactly $h(\cC)$ times.
Therefore, $\overline{\cC}$ connects a point $p$ with point $p + (m,h)$; see~\cref{fig:UCP:curved}.
If we stretch $\overline{\cC}$ to a straight line segment, it has slope~$\tfrac{h(\cC)}{m(\cC)}$ (where a line of slope~$\pm \infty$ is vertical).
Hence, all homotopy classes are given by $\bQ \cup \pm \infty$ ($m$ can be 0).
Furthermore, the \emph{standard form} of the covering curve of~$\cC$ is a line of slope~$\tfrac{h(\cC)}{m(\cC)}$ 
that goes through a grid point; see~\cref{fig:UCP:straight}.

\begin{figure}[t]
  \begin{minipage}[t]{0.45\linewidth}
    \centering
    \includegraphics[page=1]{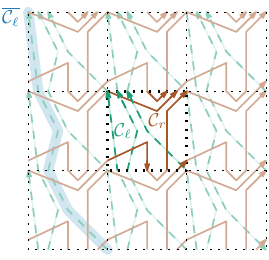}
	\subcaption{$\cy \cC_\ell$, $\br \cC_r$ and one copy $\overline{\cC_\ell}$ in the UCP.}
	\label{fig:UCP:curved}
  \end{minipage}
  \hfill
  \begin{minipage}[t]{0.45\linewidth}
    \centering
    \includegraphics[page=2,scale=0.8]{UCP.pdf}
	\subcaption{$\cy \cC_\ell$ and $\br \cC_r$ in standard form.} 
	\label{fig:UCP:straight}
  \end{minipage}
  \caption{Extract of the UCP shows two curves $\cy \cC_\ell$, $\br \cC_r$ in homotopy class~$\cy (-1,3)$ and~$\br (2,1)$ 
  cross $7 = ({\cy {-}1}) \cdot {\br 1} - {\cy 3} \cdot {\br 2}$ times.}
  \label{fig:UCP}
\end{figure}

\subsection{Cylindrical and Toroidal Graphs} % - - - - - - - - - - 
Throughout, $G$ denotes a connected (not necessarily simple) graph with vertices~$V(G)$ and with edges~$E(G)$. 
Graph~$G$ comes with a fixed, crossing-free embedding on a surface~$\Sigma$
described by giving the \emph{rotation system} (a cyclic order of edges around each vertex) that determines the faces.
An \emph{outer face} is an unbounded face on~$\Sigma$ (a cylindrical graph has two while a toroidal graphs has none); 
all other faces are called \emph{inner faces}.
We assume that the one outer face of a plane graph and the two outer faces of a cylindrical graph are part of the input.
For a toroidal graph~$G$ embedded on the unit-square flat torus $\bT$,
we additionally get the intersections of $E(G)$ with meridian ${\bl M}$ and horizon ${\re H})$.
(After local deformations, we may assume that the unique crossing of $\bl M$ and $\re H$ lies in a face of $G$, 
and that ${\bl M}, \re H$ do not go through vertices or along edges.)    
We can store these crossings via the \emph{planarization} $G^\times\!$ 
obtained by subdividing an edge $e$ with a dummy-vertex whenever $e$ is crossed by an orbit,
adding a dummy-vertex for where ${\bl M}$ crosses ${\re H}$, 
and then adding edges between dummy-vertices that are consecutive on orbits.
When we speak of finding something `in linear [quadratic] time',
then we mean `in time linear [quadratic] in the size of $G^\times$', 
i.e.\ the number of crossings of orbits with edges counts for the input-size. 

Graph $G$ is called \emph{(internally) triangulated} if all (inner) faces are triangles.   
An \emph{outer vertex} is one on an outer face; all others are \emph{inner vertices}.
A \emph{chord} of a face $f$ is an edge not on $f$ for which both endpoints are on $f$.
(In particular, if there is a vertex on $f$ that is incident to a loop, but $f$ is not the loop,
then the loop is considered a chord of $f$.)
For a toroidal graph~$G$, the \emph{cover graph}~$\tilde{G}$ of~$G$
is obtained by pasting a copy of~$G$ into each copy of~$\bT$ in the UCP.

Since we allow non-simple graphs, we have to define separating triangles in cylindrical and toroidal graphs.

\pdfbookmark[3]{Different Rectangular Duals}{GraphsDuals} 
\paragraph{Different Rectangular Duals.} % -  -  -  -  -  -  -  -  -  - 
A planar rectangular dual $\R$ has a \emph{frame}
if there are exactly four rectangles incident to the outer face of~$\R$.  
A simple planar graph~$G$ then has a planar rectangular dual with a frame if and only if it is a
\emph{properly triangulated planar (PTP) graph}~\cite{He93}:
it is planar, internally triangulated, contains no separating triangle, and the outer face is a 4-cycle.
We define rectangular duals as well as analogous necessary conditions for cylindrical and toroidal graphs
in the respective sections.

A \emph{cylindrical rectangular dual}~$\R$ of cylindrical graph~$G$ is a contact representation of~$G$ 
on a flat cylinder~$Q$ with rectangles such that 
\begin{enumerate}[(i)]
  \item no four rectangles share a point, and
  \item the union of the rectangles form a strip from the left to the right side of~$Q$,
  that is, an area bounded by two parallel lines.
\end{enumerate}
We say that~$\R$ has a \emph{frame} if there is exactly one rectangle incident to each outer face of~$\R$.
Note that we can assume without loss of generality that~$Q$ is a rectangular flat cylinder and axis-aligned; see~\cref{fig:rectangularFlatCylinder}.

Due to some particularities of cylindrical rectangular duals and their realized RELs,
the generalization of PTP graphs for cylindrical graphs has a rather specific extra property
if the graph is not planar.
A graph~$G$ is a \emph{properly triangulated cylindrical (PTC) graph} 
if~$G$ is cylindrical, internally triangulated, all loops, parallel edges, and separating triangles are non-contractible,
the outer faces have no non-consecutive occurrences of the a vertex and no chord,
and any vertex~$u$ with parallel edges to neighbors that have loops has the following degree:
\begin{itemize}
  \item If~$u$ is on an outer face,~$\deg(u) \geq 4$;
  \item if~$u$ has parallel edges to two neighbors that have loops,~$\deg(u) \geq 6$;
  \item otherwise~$\deg(u) \geq 5$.
\end{itemize}
We see in \cref{sec:cylinder} why this is needed.
Having no non-consecutive occurrences of a vertex on an outer face implies
that it is bounded by a cycle, a pair of parallel edges, or a loop.
Furthermore, we want to clarify that since $G$ may not be simple,
a separating triangle can also be spanned by a pair of parallel edges between two vertices $x$ and $y$ and a loop at $x$. 

\begin{figure}[h]
  \centering
  \includegraphics[page=4]{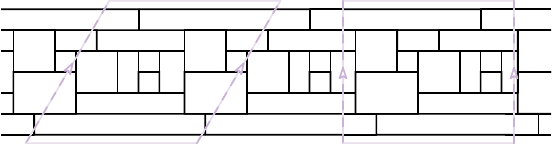}
  \caption{For cylindrical graphs, we can assume that the fundamental polygon~$Q$ is a rectangular flat cylinder.
  If otherwise, we can assume that the top and bottom sides are horizontal,
  consider an infinite chain of copies of~$Q$ glued together in the plane,
  and then pick a rectangular flat cylinder~$Q^\star$
  that covers exactly one copy of each point of~$Q$.}
  \label{fig:rectangularFlatCylinder}
\end{figure}

A \emph{toroidal rectangular dual} of a toroidal graph~$G$ is a contact representation of~$G$ 
on a flat torus~$Q$ with rectangles such that 
\begin{enumerate}[(i)]
  \item no four rectangles share a point, and
  \item the union of the rectangles fills $Q$. 
\end{enumerate}
We generally assume that~$Q$ is oriented such that the rectangles are axis aligned.
A graph~$G$ is a \emph{properly triangulated toroidal (PTT) graph}
if $G$ is toroidal, triangulated, and all loops, parallel edges, and separating triangles are non-contractible.
(The last property is also called \emph{essentially 4-connected}
since then the universal cover graph of $G$ is 4-connected~\cite{BL19}.
If the universal cover graph is simple,
then $G$ also contains no contractible loops or parallel edges.)

\pdfbookmark[3]{Regular Edge Labelings}{REL} 
\paragraph{Regular Edge Labelings.} % -  -  -  -  -  -  -  -  -  - 
Any rectangular dual~$\R$ of a graph~$G$ naturally gives rise to a 2-coloring and an orientation of the edges of~$G$:
An edge~$\set{u, v}$ is \emph{\bl blue}~[\emph{\re red}$\,$] and oriented~$(u,v)$ if the corresponding 
rectangles~$\R(u)$ and~$\R(v)$ share a horizontal [vertical] segment, 
and~$\R(u)$ is below [to the left of]~$\R(v)$.
Any inner vertex then has the \emph{REL property}: 
Going clockwise around~$v$,
there are incoming {\bl blue} edges, incoming {\re red} edges, outgoing {\bl blue} edges, and outgoing {\re red} edges.
At outer vertices some of these groups are omitted.
A \emph{regular edge labeling} (REL) is a 2-coloring and orientation of the edges of $G$ such that
the REL property holds at inner vertices, and 
some conditions at the outer faces hold.
Specifically, for a \emph{planar REL} building a frame must be feasible (cf. Chaplick \etal~\cite{CFKKRW22}).
For a \emph{cylindrical REL} the two outer faces must be bounded by directed cycles of the same color (say red),
and the REL property most hold at outer vertices except that exactly one group of blue edges is omitted.
On a torus there are no outer vertices, so in a \emph{toroidal REL} all vertices have the REL property.
We let $\cL = ({\bl L_1}, {\re L_2})$ denote a REL with~$\bl L_1$ and~$\re L_2$
representing the sets of {\bl blue} and {\re red} edges, respectively.
Let~$\blG$ and~$\reG$ denote the two (directed) subgraphs of~$G$ induced by~$\bl L_1$ and~$\re L_2$, respectively.
We write~$[2]$ as shortcut for~$\set{1, 2}$.

When speaking about a path, cycle, or walk in $L_i(G)$ (for $i \in [2]$), 
we always mean a directed path, directed cycle, or directed walk, respectively.
Note that~$\blG$ and~$\reG$ are \emph{essentially-acyclic}, i.e.\ all directed cycles are non-contractible.
For if there were, say, a {\re red} contractible cycle~$C$,
then it would split the graph into two sides, and
by the REL property there would need to be a {\bl blue} source on one side, and a {\bl blue} sink on the other;
neither would have the REL property and one could not be an outer vertex.
We direct the dual graph of~$L_i(G)$ so that all its edges cross the corresponding edges of $L_i(G)$ left-to-right.

\pdfbookmark[3]{Cycles and Curves}{Cycles} 
\paragraph{Cycles and Curves.} % -  -  -  -  -  -  -  -  -  - 
A cycle~$C$ in a directed toroidal graph defines a directed curve~$\cC$ obtained by walking along~$C$;
we therefore apply concepts such as `parallel' also to cycles.
We also consider closed walks in~$G$,
but only the special case where a cycle only touches itself.
More precisely, we say that a closed walk~$C$ is \emph{weakly simple}
if each vertex is visited at most twice and $C$ does not cross itself. 
In particular, following Chang \etal~\cite{CEX15}, for small~$\varepsilon > 0$ we can find a (simple) curve~$\cC$
that has distance at most~$\varepsilon$ to~$C$ with respect to the Fréchet distance.
Slightly abusing notation, we call $\cC$ a \emph{curve that is arbitrarily close to} $C$, 
without explicitly listing the $\varepsilon$.

For~$i \in [2]$, we call a cycle~$C$ \emph{$i$-orbital} if it is parallel to the~$i$th orbit.
For example, in \cref{fig:otherrecduals}, the {\re red} cycle $\croc{d, e , f}$ is {\re 2-orbital}, 
but none of the {\bl blue} cycles is {\bl 1-orbital}.

We later need a method to construct cycles (preferably orbital ones)
from two cycles $C_\ell,C_r$ cycles in $L_i(G)$ for some $i\in [2]$; see \cref{fig:combination}.
These cycles are principally allowed to touch each other,
but since they are cycles in $L_i(G)$ and due to the REL property, they can only touch \emph{unidirectionally}:
At any vertex $u$ in common to $C_\ell$ and $C_r$,
either at most three incident edges of~$u$ belong to $C_\ell \cup C_r$,
or the clockwise order of the four incident edges of $u$ in $C_\ell \cup C_r$ contains two incoming
edges followed by two outgoing edges.

\begin{figure}[h]
  \centering
  \begin{subfigure}[t]{0.28 \linewidth}
    \centering
    \includegraphics[page=2]{combinationExample.pdf}
	\caption{$\hati(C_r, C_\ell) = 1$}
	\label{fig:combination:one}
  \end{subfigure}
  \hspace{1cm}
  \begin{subfigure}[t]{0.25 \linewidth}
    \centering
    \includegraphics[page=1]{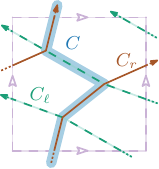}
	\caption{$\hati(C_r, C_\ell) \geq 2$}
	\label{fig:combination:many}
  \end{subfigure}
  \caption{Obtaining a weakly-simple closed walk $\bl C$ from two cycles $\cy C_\ell$ and $\br C_r$.}
  \label{fig:combination}
\end{figure}

\begin{lemma} \label{clm:combined}
  Let~$G$ be a directed essentially-acyclic toroidal graph on~$\bT$. 
  Let~$C_\ell$ and~$C_r$ be two cycles where~$C_\ell$ crosses~$C_r$ only right-to-left and at least once,
  and~$C_\ell$ and~$C_r$ may only touch unidirectionally.
  Let~$e$ be an edge of~$C_\ell$.
  
  Then in linear time we can find a weakly-simple closed walk~$C$ that contains~$e$
  and crosses both~$C_\ell$ and~$C_r$ algebraically.
  If~$\hati(C_r, C_\ell) \geq 2$ then we can choose $C$ to be simple.
  If~$m(C_\ell) < 0 < m(C_r)$ and~$h(C_\ell), h(C_r) > 0$,
  then we can choose~$C$ so that it is simple and {\bl 1-orbital}.
\end{lemma}
\begin{proof}
By the assumptions we have~$\hati(C_r, C_\ell) \geq 1$.
We first dispatch with the easier case where $\hati(C_r, C_\ell) = 1$.
Since $C_\ell$ crosses $C_r$ only right-to-left,
this means that there is exactly one crossing between them, say at subpath $C_\ell[u,v]$.
Define $C$ to be the closed walk that first traverses $C_\ell$ in its entirety (beginning and ending at $u$) and then traverses all of $C_r$.
Clearly, $C$ visits only $C_\ell[u,v]$ and where $C_\ell$ and $C_r$ touch twice, so is weakly simple, 
crosses both $C_\ell$ and $C_r$, and contains~$e$; see \cref{fig:combination:one}.
Neither of the two conditions that require $C$ to be simple can apply:
the first one contradicts the case assumption, and 
if $m(C_\ell) < 0 < m(C_r)$ and~$h(C_\ell), h(C_r) > 0$ then 
$\hati(C_r, C_\ell) = m(C_r)h(C_\ell) - h(C_r)m(C_\ell) \geq 1 \cdot 1 - 1(-1) = 2$,
again contradicting the case assumption.

Suppose now that $\hati = \hati(C_r, C_\ell) \geq 2$.
For simplicity, we assume that crossings between~$C_\ell$ and $C_r$ only occur at singleton vertices.
(Otherwise, first contract all shared paths and then expand them again in~$C$ later).
Let~$\cC_\ell$ and $\cC_r$ be the curves corresponding to $C_\ell$ and~$C_r$, respectively.
Since $C_r$ can cross $C_\ell$ only right-to-left, the number of crossings between them is exactly $\hati$
and therefore $\cC_\ell$ dissects $\cC_r$ into $\hati$ \emph{segments} and vice versa.
Let $v$ be the crossing at the start of the segment of $\cC_\ell$ containing the (potentially contracted) edge $e$.
We may assume that~$\cC_\ell$ and~$\cC_r$ are in standard form and that $v$ lies at the origin.

Refer to \cref{fig:UCP} for the following construction of a curve $\overline{\cC}$ in the UCP.
Pick some integer~$x_\ell$ with~$0 < x_\ell \leq \hati$ (we detail this choice below).
Begin at $v$, and follow~$x_\ell$ segments of~$\cC_\ell$ to reach a point~$p$ (a crossing with $\cC_r$).
Since $v$ belonged to both $\cC_\ell$ and $\cC_r$, 
following now a copy of $\cC_r$ until we reach a copy~$v'$ of~$v$, necessarily at a grid point (since $v$ was at the origin).
Let~$x_r$ be the number of segments of~$\cC_r$ between~$p$ and~$v'$.
We have~$0 < x_r \leq \hati$ since~$v' \neq p$ and since~$\cC_r$ has~$\hati$ segment-pieces.

Let~$\cC$ be the curve corresponding to $\overline{\cC}$ on $\bT$.
We verify that~$\cC$ has the desired properties:
\begin{itemize}
  \item By choice of $v$, we get that $\cC$ contains edge~$e$.
  \item Curve $\cC$ crosses $\cC_\ell$ only left-to-right since $\cC_\ell$ crosses $\cC_r$ only right-to-left.
	Also, $\cC$ crosses $\cC_\ell$ at least once, namely, 
	at the path $\cC_\ell[v,p]$ that is shared between $\cC$ and $\cC_\ell$.
	Similarly, one argues~$\hati(\cC, \cC_r) < 0$. 
 \item We can ensure that $\cC$ is simple with an appropriate choice of~$x_\ell$ as follows.

	\textit{Case 1} -- 
	We do not care whether~$\cC$ is orbital or not. In this case we use $x_\ell = 1$.
	Since $\cC_\ell$ consists of at least $\hati \geq 2$ segments, we then have $p \neq v$.
	By definition of `segments' there is no crossing of~$\cC_\ell$ and~$\cC_r$ in~$\cC_\ell(v,p)$.
	In particular therefore~$\cC_r(p,v')$ cannot visit a point that corresponds to a point in~$\cC_\ell(v,p)$.
	Since~$\cC_\ell$ and~$\cC_r$ were simple curves, therefore~$\cC$ is simple.
	See \cref{fig:combo:slanted}.

	\textit{Case 2} -- 
	We want~$\cC$ to be orbital and know~$m_\ell < 0 <m_r$ and~$h_\ell, h_r > 0$: 
	In this case we pick~$x_\ell = m_r$ (see \cref{fig:combo:orbital})
	and claim that then the choice~$x_r = {-}m_\ell$ is correct.
	To see this, recall that for~$\alpha \in \set{\ell,r}$ curve~$\cC_\alpha$ has been cut into~$\hati$ segments.
	After transformation to the standard form and looking at the UCP, 
	curve~$\cC_\alpha$ becomes a line through the vector ${m(\cC_\alpha) \choose h(\cC_\alpha)}$ and each
	of its segments becomes an equal-length piece.
	Therefore taking~$m(C_r)$ pieces of~$\cC_\ell$ and~${-}m(C_\ell)$ pieces of~$\cC_r$ brings us from the origin to point
	\[
		\frac{m(C_r)}{\hati} {m(C_\ell) \choose h(C_\ell)} + \frac{-m(C_\ell)}{\hati} {m(C_r) \choose h(C_r)} 
		= \frac{1}{\hati} {{m(C_r)m(C_\ell) -m(C_\ell) m(C_r)} \choose {m(C_r)h(C_\ell) -m(C_\ell) h(C_r)}} 
		= \frac{1}{\hati} {0 \choose \hati} 
		= {0\choose 1}.
	\]
	In particular, this brings us to a grid point, hence a copy of~$v$, 
	and since~${-}m(C_\ell) < \hati$ this choice of~$x_r$ works.
	(In fact, it is the only possible choice of~$x_r$.)
	Also, the copy of~$v$ that we have reached is at grid point~${0 \choose 1}$, 
	which means that curve~$\cC$ is in homotopy class~$(0,1)$, i.e.\ it is {\bl 1-orbital} as desired.
	Finally, the y-coordinates of the curve in the UCP are strictly increasing, and stay in the interval~$[0,1]$; 
	this cannot revisit a copy of a vertex twice since those copies would have the same y-coordinate.  
	So~$\cC$ is simple. 
\end{itemize}

\begin{figure}[t]
  \centering
  \begin{subfigure}[t]{0.28 \linewidth}
    \centering
    \includegraphics[page=4,trim=15 0 21 85,clip]{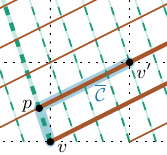}
	\caption{Choosing $x_\ell = 1$ ensures simplicity for $\hati \geq 2$.}
	\label{fig:combo:slanted}
  \end{subfigure}
  \hspace{0.5cm}
  \begin{subfigure}[t]{0.25 \linewidth}
    \centering
    \includegraphics[page=3,trim=15 0 21 85,clip]{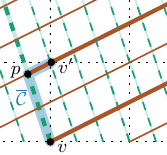}
	\caption{Choose $x_\ell = 2 = \cy m_r$ to be orbital.}
	\label{fig:combo:orbital}
  \end{subfigure}
  \hspace{0.5cm}
  \begin{subfigure}[t]{0.25 \linewidth}
  	\centering
    \includegraphics[page=5]{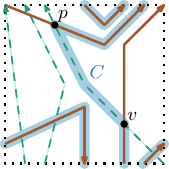}
	\caption{Translating the curve from~(a) back to $G$.}
	\label{fig:combo:translate}
  \end{subfigure}
  \caption{For the same two curves as in \cref{fig:UCP}, in \cref{clm:combined} with $\hati(C_r, C_\ell) \geq 2$,
  we find $\bl C$ (via $\bl \cC$) by following segments of $\cy C_\ell$ and $\br C_r$.}
  \label{fig:combo}
\end{figure}

So we have now (in both cases) defined a curve $\cC$ that consists of curve-pieces of~$\cC_\ell$ and~$\cC_r$,
which naturally correspond to cycle-pieces of~$C_\ell$ and~$C_r$;
let~$C$ be the closed walk that we obtain after putting these pieces together.
Before discussing why~$C$ satisfies all required properties,
we first note that we can find~$C$ directly and efficiently in~$G$, without going via the UCP.
First, we walk along $C_\ell$ and mark all vertices; then walk along $C_r$ so that we from now
on know all shared subpaths; along the way we can determine for each whether it is a touching
point or a crossing.
Next, to find~$v$, we go backward from~$e$ along~$C_\ell$ until we encounter a shared subpath that is a crossing;
$v$ is the last vertex of this subpath.
To find~$C$, we follow~$x_\ell$ cycle-pieces (defined by the crossings) of~$C_\ell$ starting at~$v$
and then follow $x_r$ cycle-pieces of~$C_r$ (thus again reaching~$v$).

Most properties of~$\cC$ transfer naturally to~$C$ since the curves were arbitrarily close to the cycles; 
in particular~$C$ contains~$e$ and crosses~$C_\ell$ and~$C_r$ algebraically but does not cross itself.
To show that $C$ is simple, we show that $C$ does not touch itself even when~$C_\ell$ and~$C_r$ touch.
Assume for contradiction that it does, say maximal subpath~$C[a,b]$ is visited again as~$C[c,d]$ and this is not a crossing.
Since $\hati(C_r, C_\ell) \geq 2$, the cycles $C_\ell$ and $C_r$  dissect $G$ into bounded regions, called \emph{patches}, in a grid like fashion.
Suppose $C$ touches itself in a patch, and let this path have the two boundary segments $C[s,x] \subset C_\ell$ and $C[x, t] \subset C_r$ on the left
as well as $C[s,y] \subset C_r$ and $C[z, t] \subset C_\ell$ on the right.
Since $C_\ell, C_r$ are simple cycles, 
the two visits of $C$ at $C[a,b]$ and $C[c,d]$ must happen once 
when $C$ follows $C_\ell$ and once when $C$ follows $C_r$.
Up to renaming $C[a,b] \subset C_\ell$, and we assume $C[a,b] \subset C_\ell[s,x]$ (the case $C[a,b] \subset C_\ell[y,t]$ is symmetric).
We know that $C[c,d] \subset C_r$, and since the touching point occurs in this patch,
we therefore must have $C[c,d] \subset C_r[x,t]$ or $C[c,d] \subset C_r[s,y]$.
\begin{figure}[h]
  \centering
  \begin{subfigure}[t]{0.28 \linewidth}
    \centering
    \includegraphics[page=3]{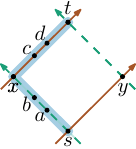}
	\caption{Case $C[c,d] \subset C_r[x,t]$.}
	\label{fig:notouching:left}
  \end{subfigure}
  \hspace{1cm}
  \begin{subfigure}[t]{0.25 \linewidth}
    \centering
    \includegraphics[page=4]{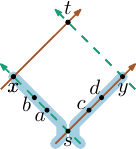}
	\caption{Case $C[c,d] \subset C_r[s,y]$.}
	\label{fig:notouching:s}
  \end{subfigure}
  \caption{For the proof of \cref{clm:combined}, we show that the cycle $\bl C$ cannot touch itself.}
  \label{fig:notouching}
\end{figure}
Assume for contradiction that $C[c,d] \subset C_r[x,t]$; see \cref{fig:notouching:left}.
Since~$C$ consists of cycle-pieces, therefore $C$ includes this entire side of the patch, $C_\ell[s,x] \cup C_r[x,t] \subset C$.   
But this directed path has repeated points (at $a,b,c,d$), so it contains a directed cycle.    
Furthermore, this directed cycle is a finite cycle in the UCP, since it resides on the boundary of one patch.
As such, it is contractible, contradicting that $G$ is essentially-acyclic.   
So the only way in which cycle $C$ could touch itself is that $C[c,d] \subset C_r[s,y]$; see \cref{fig:notouching:s}.
In this case, $C$ uses the cycle-pieces $C_\ell[s,x]$ as well as $C_r[s,y]$, 
and we must have used two pieces of $C_\ell$ and $C_r$ that go away from a common start point.   
But this does not happen since $C_\ell(v,p)$ and $C_r(p,v')$ do not have a vertex in common 
(cycle $C$ would cross itself here),
and $p$ and $v$ do not correspond to the same vertex of $G$ by $\hat{i}\geq 2$.
\end{proof}

We also use results symmetric to \cref{clm:combined} for~$i = 2$ and/or $e \in C_r$.

\section{Rectangular Dual on the Torus} % ------------------------
\label{sec:torus}
In this section, we give an algorithm to decide whether a given toroidal graph~$G$ on the flat torus~$\bT$
admits a toroidal rectangular dual, 
assuming that a toroidal REL $\cL$ has been specified.
To this end, we first classify toroidal RELs and then describe construction algorithms.

\subsection{Toroidal RELs: the Good, the Bad, and the Ugly} % - - - - - - - - - - - 
\label{sec:GoodBadUgly}
We classify $\cL$ based on whether it can be realized by a rectangular dual 
on the rectangular or slanted flat torus or not at all.

\paragraph{A ``Bad'' REL.}
For $i \in [2]$, an edge $e \in L_i(G)$ is called \emph{lonely} if it does not belong to any cycle of~$L_i(G)$.
We call $\cL$ \emph{unrealizable} if it has a lonely edge, and \emph{realizable} otherwise. 
As we show in \cref{sec:toroidalRD}, only a realizable REL can come from a rectangular~dual.

\paragraph{A ``Good'' REL.}
For $i \in [2]$, an edge $e \in L_i(G)$ is called \emph{orbital} 
if it lies on an $i$-orbital cycle of $L_i(G)$.
We call $\cL$ \emph{$i$-orbital} if all edges in $L_i(G)$ are orbital,
and \emph{orbital} if it is {\bl 1-orbital} and {\re 2-orbital}.   
(The REL in \cref{fig:otherrecduals} is not {\re 2-orbital}: the {\re red} edge $(d,i)$
does not belong to a {\re 2-orbital} {\re red} cycle.)
Note that orbital implies realizable.
Moreover, we show below that `orbital' is equivalent to `realizable by a rectangular dual on the rectangular flat torus'.

\paragraph{An ``Ugly'' REL.}
Presuming $\cL$ is realizable, we call the REL \emph{slanted} if, 
for some $i \in [2]$, graph~$L_i(G)$ contains an edge that is not $i$-orbital.
Such a REL is realizable by a toroidal rectangular dual, but not on the rectangular flat torus.

\paragraph{Recognition.} 
We can test whether $\cL$ is realizable in linear time by computing
the strongly connected components of $\blG$ and $\reG$~\cite{Tar72} 
and checking that all edges belong to one (lonely edges would not).
Assume from now on that $\cL$ is realizable, otherwise we are done.
The following observation becomes useful later:

\begin{observation} \label[observation]{clm:reverseparallel}   
  If a toroidal REL $\cL$ is realizable, and if, for some $i \in [2]$, 
  two cycles $C, C'$ in $L_i(G)$ satisfy $\hati(C, C') = 0$,
  then $C$ and $C'$ are parallel.
\end{observation}
\begin{proof}
Cycles $C, C'$ are non-contractible since toroidal RELs have no contractible cycles.
So by $\hati(C, C') = 0$ they are either parallel or reverse-parallel.
Assume for contradiction the latter.
Then it is easy see that, due to the REL property and $C$ and $C'$ being cycles,
they do neither intersect nor touch.
(Crossings would force a contractible cycle or, like touching, a curve to self intersect.)
Hence, $C$ and $C'$ dissect the flat torus into two areas, 
one to the right and one to the left of both. 
Consider the \textit{other} graph~$L_{3-i}(G)$.  
Any vertex~$v$ on~$C$ has outgoing edges in~$L_{3-i}(G)$, 
and by the REL property these lead to the right of~$C$.
Therefore, the curve $D$ that is slightly
to the right of~$C$ corresponds to a directed cycle in the dual graph of~$L_{3-i}(G)$.
Similarly, the curve $D'$ that is slightly to the right of $C'$ 
corresponds to a directed cycle in the dual graph of~$L_{3-i}(G)$.
Curves~$D$ and $D'$ are disjoint and together form a directed cut in~$L_{3-i}(G)$ between the two areas; see \cref{fig:noReverseParallel}.
The edges in this cut are lonely, a contradiction to~$\cL$ being realizable.
\end{proof}
\begin{figure}[b]
	\centering    
  \begin{minipage}[t]{0.46\linewidth}
	\centering    
    \includegraphics[page=1]{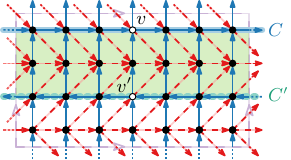}
	\subcaption{$\bl C$ and~$\cy C'$ can be reverse-parallel\ldots}
  \end{minipage}
  \hfill
  \begin{minipage}[t]{0.5\linewidth}
	\centering    
    \includegraphics[page=2]{noReverseParallel}
	\subcaption{\ldots but then the REL has lonely edges.}
  \end{minipage}
  \caption{No realizable REL has reverse-parallel cycles.}
\label{fig:noReverseParallel}
\end{figure}

We now work towards the (non-trivial) algorithm to test whether $\cL$ is orbital,
which depends on the embedding on~$\bT$ (and not just on the rotation system).
We call a cycle~$C$ in~$L_i(G)$ a \emph{left-first} [\emph{right-first}$\,$] cycle
if the leftmost [rightmost] outgoing edge at each vertex of~$C$ is in~$C$.
Due to this preference of outgoing edges, the following is easily shown:

\begin{observation} \label[observation]{clm:LeftRightProperties}
  Let $C_\ell$ and $C_r$ be a left-first and right-first cycle of $L_i(G)$ for some $i \in [2]$ of a toroidal REL. Then:
  \begin{enumerate}[(1)] 
	\item No cycle in $L_i(G)$ can cross $C_\ell$ right-to-left or cross $C_r$ left-to-right.
		\label{clm:leftRight:howtocross}
	\item All left-first [right-first] cycles in $L_i(G)$ are parallel with $C_\ell$ [$\,C_r$].
		\label{clm:leftRight:parallelLeft}
	\item If $C_\ell$ and $C_r$ are parallel, then all cycles in $L_i(G)$ are parallel to both.
		\label{clm:leftRight:parallelBoth}
  \end{enumerate}
\medskip
\end{observation}
\begin{proof} 
Claim~(1) holds for~$C_\ell$ since no vertex on~$C_\ell$ has outgoing edges to the left of~$C_\ell$.
Claim~(2) holds for~$C_\ell$ because $C_\ell$ could not cross a left-first cycle $C_\ell'$ right-to-left
by (1) or left-to-right by~(1) applied to $C_\ell'$ with respect to~$C_\ell$.
Symmetric arguments show~(1) and~(2) for~$C_r$.

To see (3), fix an arbitrary cycle~$C$ in~$L_i(G)$. 
We have~$\hati(C, C_\ell) \geq 0$ by (1),~$\hati(C, C_r) \leq 0$ by (1),
and~$\hati(C, C_\ell) = \hati(C, C_r)$ since~$C_\ell, C_r$ are parallel.
This is possible only if $\hati(C, C_\ell) = 0 = \hati(C, C_r)$,
which by \cref{clm:reverseparallel} gives the result.
\end{proof}

For~$i \in [2]$, we can compute a left-first cycle~$C_\ell$
and a right-first cycle~$C_r$ in~$L_i(G)$ in linear time with graph traversals.
Then we compute~$\hati(A, B)$ for any distinct~$A, B \in \set{C_\ell, C_r, {\bl M}, {\re H}}$,
where ${\bl M}$ and ${\re H}$ are the meridian and the horizon of the embedding on $\bT$.
For~$i = 1$, a cycle $C$ in $L_i(G)$ is called~\emph{$i$-\upward} if~$h(C) > 0$ and cycles $C_\ell, C_r$
are called~\emph{$i$-\enclosing} if either~$m(C_\ell) = 0 = m(C_r)$ or~$m(C_\ell) < 0 < m(C_r)$.
(By~\cref{clm:LeftRightProperties}(\ref{clm:leftRight:parallelLeft})
this is independent of the choice of $C_\ell$ and $C_r$.) 
For $i = 2$, the definition is symmetric; exchange $h$ and~$m$.

\begin{observation} \label[observation]{clm:orbitalImpliesUpward}
For an orbital REL~$\cL$ and $i \in [2]$, any cycle~$C$ in~$L_i(G)$ is~$i$-\upward.
\end{observation}
\begin{proof}
We only show the claim for~$i = 1$ 
(i.e.\ $\bl C$ is blue/in~$\blG$ and we must show~$h({\bl C}) > 0$); 
the proof for~$i = 2$ is similar.
Fix an arbitrary vertex~$v \in \bl C$, and let~$e$ be an incoming red edge at~$v$.   
Since~$\cL$ is orbital, there exists a {\re 2-orbital} cycle~$\re C_e$ containing~$e \in \reG$.
By the REL property~$\re C_e$ crosses~$\bl C$ at~$v$, 
and any other crossing of the red cycle~$\re C_e$ with the blue cycle~$\bl C$ (if any) is also left-to-right,
so $\hati({\re C_e},{\bl C})>0$.
Since~$(m({\re C_e}), h({\re C_e})) = (1, 0)$,
we have~$0 < \hati({\re C_e}, {\bl C}) = m({\re C_e})h({\bl C}) - h({\re C_e})m({\bl C}) = h({\bl C})$.
\end{proof}

Using \cref{clm:combined,clm:orbitalImpliesUpward}, we can show the following characterization.
We use repeatedly the observation that a $1$-\upward cycle~$C$ with $m(C) = 0$ is 1-orbital.

\begin{proposition} \label[proposition]{clm:orbitalCharacterization}
  A realizable REL~$\cL$ is orbital if and only if, for both~$i \in [2]$, 
  a left-first cycle~$C_\ell$ and right-first cycle~$C_r$ in~$L_i(G)$
  are~$i$-\upward and~$i$-\enclosing. 
\end{proposition}
\begin{proof}
First assume that $\cL$ is orbital. 
By \cref{clm:orbitalImpliesUpward} both~$C_\ell$ and~$C_r$ are~$i$-\upward, 
so we must only prove that they are~$i$-\enclosing,
i.e.\ (assuming again~$i = 1$) that~$m(C_\ell) \leq 0 \leq m(C_r)$ with either none or both inequalities tight.
Since~$\cL$ is orbital and has at least one {\bl blue} edge, 
we have a {\bl 1-orbital} cycle~$\bl C_o$ in the blue graph~$\blG$, so~$m({\bl C_o}) = 0$.  
If~$\bl C_o$ is parallel to~$C_\ell$ then~$m(C_\ell) = m({\bl C_o}) = 0$.   
Otherwise,~$\bl C_o$ crosses~$C_\ell$,
necessarily only left-to-right since~$C_\ell$ is a left-first cycle, 
and so~$0 < \hati({\bl C_o}, C_\ell) = m({\bl C_o})h(C_\ell) - m(C_\ell)h({\bl C_o}) = -m(C_\ell)$.   
Either way $m(C_\ell) \leq 0$, and symmetrically one argues that~$m(C_r) \geq 0$ since~$\bl C_o$ crosses it only right-to-left.

So~$m(C_\ell) \leq 0 \leq m(C_r)$, and we are done unless exactly one of the inequalities is tight.
Assume for contradiction that~$m(C_\ell) = 0 < m(C_r)$ (the other case is symmetric).    
Then~$C_r$ is not parallel to~$C_\ell$, cannot be reverse-parallel to it, 
and therefore must cross~$C_\ell$, necessarily left-to-right.
Let~$e$ be an edge of~$C_r$ incoming at such a crossing, i.e.\ 
the head of~$e$ belongs to both~$C_\ell$ and~$C_r$, but~$e \not\in C_\ell$.
Since the crossing is left-to-right,~$e$ is left incoming to~$C_\ell$. 
Since~$\cL$ is~{\bl 1-orbital}, there exists a {\bl 1-orbital} cycle~$\bl C_e$ through~$e$ in the blue graph~$\blG$. 
Cycle~$\bl C_e$ is on~$C_\ell$ at the head of~$e$, may run along it for a while, 
but eventually must leave it (since~$e \not\in C_\ell$), 
and it can leave~$C_\ell$ only to the right.  
Therefore cycle~$\bl C_e$ crosses~$C_\ell$ left-to-right at the head of~$e$,
and it cannot ever cross it right-to-left.   
So~$\hati({\bl C_e}, C_\ell) > 0$, which contradicts that both~$\bl C_e$ and~$C_\ell$ are~{\bl 1-orbital} by $m(C_\ell) = 0$. 

\medskip
For the other direction, we again only consider $i=1$ and assume that 
we are given a left-first cycle~$C_\ell$ 
and right-first cycle~$C_r$ in~$L_i(G)$ that are~$i$-\upward ($h(C_\ell), h(C_r) > 0$) 
and~$i$-\enclosing, and we show that $\cL$ is then $i$-orbital.
Fix an arbitrary edge~$e \in L_i(G)$; we want to find an~$i$-orbital cycle through~$e$ in~$L_i(G)$.
Let~$C_e$ be some cycle through~$e$ in~$L_i(G)$, which exists since~$\cL$ is realizable.

If $m(C_\ell) = 0 = m(C_r)$ then $C_\ell$ and $C_r$ are 1-orbital
and $C_e$ is parallel to both by \cref{clm:LeftRightProperties}(\ref{clm:leftRight:parallelBoth}); 
so $C_e$ itself is 1-orbital and we are done.
So we may assume that~$m(C_\ell) < 0 < m(C_r)$.   
We first show $C_e$ is 1-\upward.  
This clearly holds if~$C_e$ is parallel with~$C_r$ or $C_\ell$.
Furthermore, $C_e$ cannot be reverse-parallel with either, say $C_r$,
as this would force it to cross $C_\ell$ right-to-left.
So assume that $C_e$ crosses both.
If~$m(C_e) \geq 0$, then (since any crossing of~$C_e$ with~$C_r$ is right-to-left)
we have~$0 > \hati(C_e, C_r) = m(C_e)h(C_r) - m(C_r)h(C_e)$ or~$h(C_e) > m(C_e)h(C_r) / m(C_r) \geq 0$.  
Similarly~$m(C_e) \leq 0$ implies $h(C_e) > 0$ since $\hati(C_e,C_\ell) > 0$. So $C_e$ is 1-\upward.
If~$m(C_e) = 0$, then this makes~$C_e$ itself 1-orbital and we are done.
If~$m(C_e) > 0$ [$m(C_e) < 0$], then~$C_\ell$ and~$C_e$ [$C_e$ and~$C_r$] satisfy all preconditions for \cref{clm:combined}, 
and using it we can find an~$i$-orbital cycle that contains edge~$e$.
\end{proof}

So we can read whether $\cL$ is orbital from the algebraic crossing numbers of ${\bl M}, {\re H}$ 
with the left-first and right-first cycles in $L_i(G)$ for $i\in [2]$,
and have:

\begin{theorem} \label{clm:RDflat:linear}
  Given a PTT graph $G$ on the flat torus $\bT$,
  it can be decided in linear time whether a REL of~$G$ is unrealizable, orbital or slanted.
\end{theorem}
\begin{proof}
First test whether $\cL$ is realizable, then for $i \in [2]$ compute some left-first and
right-first cycles $C_\ell,C_r$ in $L_i(G)$ and their homotopy classes.
From this, we can immediately read whether they are $i$-\upward and $i$-\enclosing.
By \cref{clm:orbitalCharacterization} we therefore know whether $\cL$ is orbital.
If $\cL$ is neither unrealizable nor orbital then it is slanted.
\end{proof}

We want to point out that for \cref{clm:orbitalCharacterization}
(and specifically \cref{clm:orbitalImpliesUpward})
it is required that the REL is orbital for \emph{both}~$i = 1$ and~$i = 2$.
\cref{fig:torus:notSufficientForOrbital} shows a realizable REL that is 1-orbital but not~$1$-\upward. 

\begin{figure}[tbh]
  \centering
  \includegraphics[page=1]{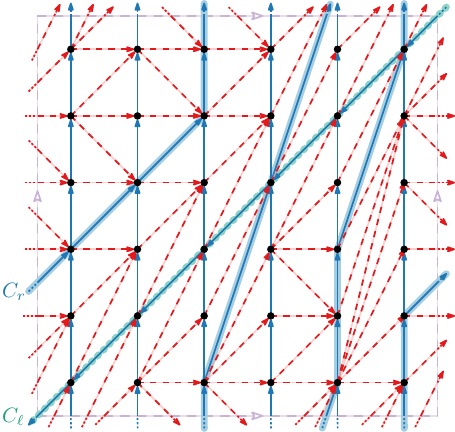}
  \caption{This toroidal REL is {\bl 1-orbital} but not {\re 2-orbital}.
  One can verify that every {\bl blue} edge is {\bl 1-orbital}
	(use vertical paths for vertical edges, and combine vertical paths
	with exactly one edge each of ${\cy C_\ell}$ and ${\bl C_r}$ for edges within
	${\cy C_\ell}/{\bl C_r}$).
  But the left-first cycle $\cy C_\ell$ of $\blG$ goes downward and leftward, so it has homotopy-class~$({-}1,{-}1)$  
  and is not 1-\upward.}
  \label{fig:torus:notSufficientForOrbital}
\end{figure}

\subsection{From Realizable REL to Toroidal Rectangular Dual} % - - - - - - - - - - - 
\label{sec:toroidalRD}
We now show that what we defined as \textit{realizable RELs} is indeed a characterization,
that is, they are exactly those RELs 
that can arise from a toroidal rectangular duals.
That realizability is necessary is easy to show:
If graph~$G$ has a toroidal rectangular dual~$\R$,
then we can for each edge~$e$ in~$L_i(G)$ easily extract a cycle in $L_i(G)$ containing~$e$.
Simply start on a line perpendicular to a contact that represents $e$ and see where it leads~in~$\R$.

\begin{lemma} \label{clm:necessary}
  If a PTT graph~$G$ has a rectangular dual~$\R$ on the flat torus~$Q$, 
  then the corresponding REL~$\cL$ is realizable.
  If $Q$ is a rectangle, then $\cL$ is orbital. 
\end{lemma}
\begin{proof}
Consider the universal covering plane of $Q$, i.e., fill the plane with copies of $Q$.
By pasting~$\R$ into each copy of~$Q$,  we can create a periodic rectangular dual~$\tilde{\R}$ of the cover graph $\tilde{G}$.
We only show that every {\bl blue} edge~$e$ belongs to a cycle $C_e$ (and~$C_e$ can be chosen {\bl 1-orbital} if $Q$ is a rectangle);
the proof for {\re red} edges works analogously.
Edge~$e$ is represented by a horizontal segment~$s_e$ in~$\R$; 
let~$\ell_e$ be a vertical line in~$\tilde{\R}$ that intersects one copy~$\tilde{s_e}$ of~$s_e$ 
and that does not intersect any left or right side of a rectangle.
Since~$\tilde{\R}$ is periodic, line~$\ell_e$ must intersect another copy~$\hat{s}_e$ of~$s_e$ above~$\tilde{s}_e$.
Line~$\ell_e$ never runs along a vertical side of a rectangle, so intersects only horizontal segments (representing edges in~$\blG$).
So walking along~$\ell_e$ from~$\tilde{s_e}$ to~$\hat{s_e}$
defines a closed walk in~$\blG$ that contains~$e$. 
Within this closed walk, we can find a cycle~$\bl C_e$ that contains~$e$, so~$e$ is not lonely.

If~$Q$ is an axis-aligned rectangle, then covering curves of the meridian~${\bl M}$ correspond to vertical lines.
Cycle~$\bl C_e$ was derived from the vertical line~$\ell_e$,
hence its covering curves can be made vertical in a suitable drawing of $G$ that overlays $\R$.
Therefore, $\bl C_e$ is parallel to $M$ and hence orbital as~desired.
\end{proof}

The difficulty lies in showing that this necessary condition is also sufficient. 
We first give an outline.
Given a PTT graph~$G$ and a realizable REL~$({\bl L_1}, {\re L_2})$ of~$G$,
the idea is to adapt He's algorithm~\cite{He93}, 
which is based on the observation that 
the faces of~$\blG$ [$\reG$] are in bijection with the maximal vertical [horizontal] segments in~$\R$; see~\cref{fig:torus:rel}.
To compute the coordinates of these segments,
the algorithm splits (for $i \in [2]$) the outer face of $L_i(G)$ in two;
in a planar graph the dual of $L_i(G)$ then becomes acyclic and coordinates 
can be read from longest paths in these dual graphs.
The main obstacle in the toroidal case is that the dual of $L_i(G)$ is not so easily made acyclic.
We instead split the entire graph $L_i(G)$ along some weakly-simple closed walk $C$ 
such that the result (a cylindrical graph denoted $L_i(G) \para C$) has an acyclic dual.
For this, $C$ must cross every directed cycle in $L_i^\star(G)$; 
we therefore call it a \emph{feedback closed walk}.
We can then extend He's algorithm to compute the coordinates.

\begin{figure}[tbh]
  \centering
  \begin{minipage}[t]{0.45 \linewidth}
	\centering
   	\includegraphics[page=1]{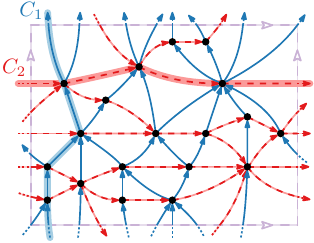}
	\subcaption{PTT graph $G$ with orbital REL $\set{{\bl L_1}, {\re L_2}}$}
	\label{fig:torus:rel:g}
  \end{minipage}
  \qquad
  \begin{minipage}[t]{0.45 \linewidth}
	\centering
    \includegraphics[page=2]{torusRELV1}
	\subcaption{$\bl L_1 \para C_1$ and $\og L_1^\star(G)$}
	\label{fig:torus:rel:blue}
  \end{minipage}
  
  \vspace{0.5cm}
  
  \begin{minipage}[t]{0.45 \linewidth}
	\centering
   	\includegraphics[page=3]{torusRELV1}
	\subcaption{$\re L_2 \para C_2$ and $\gr L_2^\star(G)$}
	\label{fig:torus:rel:red}
  \end{minipage}
  \qquad
  \begin{minipage}[t]{0.45 \linewidth}
	\centering
    \includegraphics[page=4]{torusRELV1}
	\subcaption{Rectangular dual $\R$}
	\label{fig:torus:rel:recdual}
  \end{minipage}
  \caption{Given an orbital REL, we split $\blG$ and $\reG$ along orbital feedback cycles
  and then use He's algorithm to obtain a rectangular dual $\R$ on the flat rectangular~torus.}
  \label{fig:torus:rel}
\end{figure}

\paragraph{Feedback closed walk.}
Fix $i \in [2]$. A \emph{feedback closed walk} is a non-contractible weakly-simple closed walk~$C$ of $L_i(G)$ that
crosses all cycles in the dual of $L_i(G)$ and that visits each vertex of $L_i(G)$ at most twice.  
Define graph $L_i(G) \para C$ by \emph{splitting $L_i(G)$ along~$C$} as follows:

\textit{Case 1 -- $C$ is a cycle:}
Duplicate~$C$ into two copies~$C_s$ and~$C_t$;
for a vertex~$v \in V(C)$, let the copies be~$v_s$ and~$v_t$, respectively.
Distribute each edge incident to~$v$ but not on~$C$ accordingly to~$v_s$ and~$v_t$
based on whether it attaches to~$C$ at~$v$ from the right or the left, respectively.

\textit{Case 2 -- $C$ is weakly simple:}
Let $C[u,v]$ be a subpath of $C$ that visited twice; see~\cref{fig:torus:single:g}.
As before, duplicate $C$ into two copies~$C_s$ and~$C_t$ and distribute edges accordingly,
which results in three copies of each vertex on $C[u,v]$ instead of four copies; see~\cref{fig:torus:single:split}.
(Imaging cutting through a drawing of $L_i(G)$ with edges of non-zero thickness.
Then the cut moves through $C[u,v]$ twice, creating three copies.)  

Either way, graph $L_i(G) \para C$ is obtained by splitting a toroidal graph along a non-contractible curve, 
so it is cylindrical with its two outer faces incident to $C$
(the two outer faces share subpaths whenever $C$ visits such a path twice).
So if $C$ crosses all directed cycles of $L_i^\star(G)$, then the dual of~$L_i(G)\para C$ is acyclic.  
So a crucial step is how to find a feedback closed walk, preferably a simple and orbital~one.

\begin{figure}[tbh]
  \centering
  \begin{minipage}[t]{0.4 \linewidth}
	\centering
   	\includegraphics[page=1]{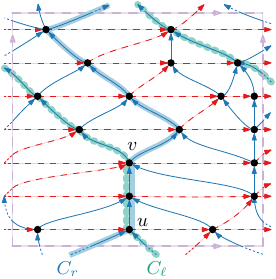}
	\subcaption{PTT graph~$G$ with~$\set{{\bl L_1}, {\re L_2}}$}
	\label{fig:torus:single:g}
  \end{minipage}
  \hfill
  \begin{minipage}[t]{0.55 \linewidth}
	\centering
    \includegraphics[page=2]{singleSelfTouching}
 	\subcaption{$\bl\blG \para C_1$ and~$\og L_1^\star(G)$}
 	\label{fig:torus:single:split}
  \end{minipage}
  \caption{The feedback cycle $C_1$ (obtained from $C_\ell$ and $C_r$) of this slanted REL $\set{{\bl L_1}, {\re L_2}}$ is weakly simple.
  Hence, $\bl\blG \para C_1$ contains four copies of the vertices on the path $C_1[u, v]$ shared by $C_\ell$ and $C_r$.}
  \label{fig:torus:single}
\end{figure}

\begin{lemma} \label{clm:feedbackCycle}
  Let~$\cL$ be a realizable REL of a PTT graph~$G$. 
  For~$i \in [2]$, a feedback closed walk~$C$ in~$L_i(G)$ exists and can be found in linear time.
  Moreover, if~$\cL$ is~orbital, then~$C$ can be chosen $i$-orbital and simple.
\end{lemma}
\begin{proof}
Fix~$i \in [2]$, and let $D$ be an arbitrary cycle in the dual of~$L_i(G)$; see \cref{fig:feedbackcycle:D}.
We now choose $C$ (independently of $D$) and then argue that $C$ crosses $D$.
For this, we need two observations about $D$.   
First, $\hati(D, C') \geq 0$ for \textit{any} cycle $C'$ in~$L_i(G)$ since dual edges are directed to cross only left-to-right over edges of~$L_i(G)$.
Second, there exists \textit{some} cycle $C_D$ in~$L_i(G)$ with $\hati(D, C_D) > 0$; for example we can take any edge $e^\star$ in $D$,
its corresponding edge $e$ in~$L_i(G)$, and let $C_D$ be a cycle in~$L_i(G)$ through $e$ (this exists since $\cL$ is realizable).

\begin{figure}[t]
  \centering
  \begin{minipage}[t]{0.4 \linewidth}
	\centering
   	\includegraphics[page=5]{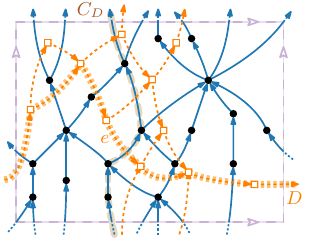}
	\subcaption{A cycle $\og D$ in $\og L_1^\star(G)$ and a cycle $\br C_D$ in $\blG$ that crosses it.}
	\label{fig:feedbackcycle:D}
  \end{minipage}
  \qquad
  \begin{minipage}[t]{0.5 \linewidth}
	\centering
    \includegraphics[page=6]{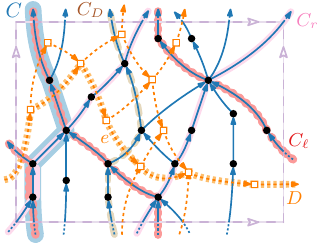}
	\subcaption{The left-first and right-first cycles $C_\ell$ and $C_r$ have $\hati (C_\ell,C_r)\neq 0$;
	the resulting feedback closed walk is $\bl C$.}
	\label{fig:feedbackcycle:C}
  \end{minipage}
  \caption{Constructing a feedback closed walk $\bl C$ for $\blG$ in the proof of \cref{clm:feedbackCycle}.}
  \label{fig:feedbackcycle}
\end{figure}

Now fix an arbitrary left-first cycle~$C_\ell$ and right-first cycle~$C_r$ in~$L_i(G)$ and consider two cases; see \cref{fig:feedbackcycle:C}.
First, if $\hati(C_\ell, C_r) = 0$, then $C_\ell,C_r$ are parallel by \cref{clm:reverseparallel} 
and $C_D$ is parallel to~$C_r$ by \cref{clm:LeftRightProperties}(\ref{clm:leftRight:parallelBoth}). 
Therefore, $\hati(D, C_r) = \hati(D, C_D) > 0$ and~$C_r$ crosses $D$, and we can use $C_r$ as (simple) feedback cycle.
Second, suppose $\hati(C_\ell, C_r) \neq 0$.
By \cref{clm:LeftRightProperties}(\ref{clm:leftRight:howtocross}) cycle $C_\ell$ crosses $C_r$ only right-to-left, 
so use \cref{clm:combined} to obtain a closed walk $C$ that crosses both $C_\ell$ and~$C_r$ algebraically.
In fact, $\hati(C, C_\ell) < 0 < \hati(C, C_r)$ by \cref{clm:LeftRightProperties}(\ref{clm:leftRight:howtocross}).   
If $C$ were parallel with~$D$, then~$\hati(D, C_\ell) = \hati(C, C_\ell) < 0$; impossible.
If $C$ were reverse-parallel with~$D$, then $\hati(D, C_r) = -\hati(C, C_r) < 0$; impossible.
So $C$ must cross $D$, and all other properties of feedback closed walks hold for $C$ by \cref{clm:combined}.  
\end{proof}

\paragraph{Computation of Coordinates.}
For an orbital REL, the computation of coordinates is now nearly identical 
to the algorithm by He~\cite{He93}, after splitting the~graph.

\begin{theorem} \label{clm:torus:relToRecDual}\label{clm:torus:orbital}
  For a PTT graph~$G$ on $\bT$ with an orbital REL~$\cL$, 
  a toroidal rectangular dual~$\R$ of~$G$ that realizes~$\cL$ 
  and lies on a rectangular flat torus can be computed in linear time.
\end{theorem}
\begin{proof}
We adapt He's algorithm; see also \cref{fig:torus:rel}. 
Let~$i \in [2]$.
First, compute an orbital feedback \textit{cycle}~$C_i$ of~$L_i(G)$ with \cref{clm:feedbackCycle},
the cylindrical graph~$L_i(G) \para C_i$, 
and let~$L_i^\star(G)$ be the dual of~$L_i(G) \para C_i$ with source~$s_i$ and sink~$t_i$
corresponding to the outer faces.   
Second, compute what He calls a \emph{consistent numbering} on~$L_i^\star(G)$:
for a face~$f$, set~$d_i(f)$ to be the length of a longest path from~$s_i$ to~$f$.

As the third step, use~$\bl d_1$ and~$\re d_2$ to obtain x- and y-coordinates.
Note that graph~$L_i(G) \para C_i$, $i \in [2]$, is bimodal,
that is at every vertex $v$ the incoming edges are consecutive around $v$, as are the outgoing edges, 
and both sets are non-empty by the REL property and since we split along a cycle.   
We let the \emph{faces to the left and right of~$v$} (denoted $\mathrm{left}_i(v)$ and~$\mathrm{right}_i(v)$)
be the faces where we transition from incoming to outgoing edges and vice versa.
For~$i = 1$ [$i = 2$] and~$v \not\in V(C_i)$,
the left and right x-coordinate [top and bottom y-coordinate] of~$\R(v)$
are given by~$d_i(\mathrm{left}_i(v))$ and~$d_i(\mathrm{right}_i(v))$, respectively.
For~$v \in V(C_i)$, use~$v_s$ and~$v_t$ for the two copies of~$v$ on the outer faces of~$L_i(G)\para C_i$.
Then the left and right x-coordinate [top and bottom y-coordinate] 
of~$\R(v)$ are given by~$d_i(\mathrm{left}_i(v_t))$ and~$d_i(\mathrm{right}_i(v_s))$, respectively.

As canonical parallelogram~$Q$ we use a rectangle with width~${\bl d_1}(t_1) - 1$ and height~${\re d_2}(t_2) - 1$,
and with the lower-left corner at~$0.5 \choose 0.5$.
Note that all computed coordinates thus lie in~$Q$; see \cref{fig:torus:rel:recdual}.
It remains to show that~$\R$ is indeed a rectangular dual.

Note that the left and right side of~$Q$ corresponds to~$\bl C_1$
and the top and bottom side of~$Q$ corresponds to~$\re C_2$.
For~$i = 1$ [$i = 2$] and~$v \in V(C_i)$, 
we have to show that the first x-coordinate [y-coordinate] is lower than the second.
Since~$\cL$ is orbital, there is a {\re red} [{\bl blue}] path~$P$ from~$v_s$ to~$v_t$ in~$\reG$ [$\blG$].
By the REL property, there is a path parallel to~$P$ in~$L_i^\star(G)$
that goes from~$\mathrm{right}_i(v_s)$ to~$\mathrm{left}_i(v_t)$.
As~$L_i^\star(G)$ is acyclic,~$d_i(\mathrm{right}_i(v_s)) < d_i(\mathrm{left}_i(v_t))$. 
That they align vertically [horizontally] follows, as with He, from~$\re d_{2}$ [$\bl d_1$]. 
Hence,~$\R$ is the desired rectangular dual
and, since~$\bl C_1$ is parallel to~$\bl M$ and~$\re C_2$ is parallel to~$\re H$,
we get that~$Q$ also respects the embedding of~$G$ on~$\bT$.
\end{proof}

For a slanted REL, computing rectangle-coordinates is significantly more complicated 
since the consistent numberings~$\bl d_1$ and~$\re d_2$ do not directly translate
to coordinates in the slanted flat torus.
By tracking how~$\bl M$ and~$\re H$ intersect the feedback cycle~$L_i(G)$ for $i\in [2]$,
we can piece together a rectangular dual that lies in a slanted flat torus corresponding to~$({\bl M}, {\re H})$. 

\begin{theorem} \label{clm:torus:slanted}
  For a PTT graph~$G$ on $\bT$ with a slanted REL~$\cL$,
  a toroidal rectangular dual~$\R$ of~$G$ that realizes~$\cL$ can be computed in quadratic time.
  The boundaries of the flat torus that defines~$\R$ correspond to curves that are parallel to $\bl M$ and $\re H$ of $\bT$.
\end{theorem}
\begin{proof}
A full example illustrating the construction is given in \cref{sec:example}.
Exactly as in the proof of \cref{clm:torus:orbital},
we compute, for~$i \in [2]$, a feedback closed walk~$C_i$ (see \cref{fig:example:feedback})
and the consistent numberings~$d_i(\cdot)$ in~$L_i^\star(G)\para C_i$.
Note that $C_1$ induces a cycle in the dual of $L_2(G)$ and thus~$C_1$ must cross~$C_2$. 
Furthermore, due to the REL property, $C_1$ and~$C_2$ cannot touch, only cross, and any crossing is right-to-left.
Hence, $C_1$ and $C_2$ dissect $G$ into bounded regions, called \emph{patches}, in a grid like fashion;
see \cref{fig:example:patches}.
Each patch~$P$ corresponds to a finite subgraph of the cover graph~$\tilde G$ and hence is a planar graph.
As in the proof of \cref{clm:torus:orbital}, 
we can extract a rectangular dual~$\R_P$ for patch~$P$.
We use here the consistent numberings from~$L_i^\star(G)$ (\textit{not} the ones in the dual of~$P$),
so~$\R_P$ is a~$(d_2(t_2){-}1) \times (d_1(t_1){-}1)$-rectangle.
See \cref{fig:example:duals,fig:example:patchduals}.
These rectangular duals of patches can be combined to get
a periodic infinite rectangular dual~$\tilde{\R}$ of the cover graph~$\tilde{G}$.   
Namely, if a patch~$P'$ is adjacent to~$P$ in the sense that their boundaries share a walk-piece
(say~$C_1[a,b]$ has~$P$ to its left and~$P'$ to its right), 
then for any~$v$ on~$C[a,b]$ the two corresponding rectangles~$\R_P(v)$ and~$\R_{P'}(v)$ 
have the exact same y-coordinates at the top and bottom, 
because these only depend on~$d_2(\cdot)$ for two faces in~$\reG$ and not on~$P$ or~$P'$.   
In particular, by translating~$\R_{P'}$ by~$d_1(t_1)$
(hence placing it next to~$\R_P$, with~$P$ on the left) 
we obtain a rectangular dual of~$P \cup P'$.
If we tiled the entire plane with~$d_2(t_2) \times d_1(t_1)$-rectangles,
and placed the rectangular duals of the patches in them as dictated by adjacencies among them, 
this would give~$\tilde{\R}$; see \cref{fig:example:tiles}.

\subparagraph{Finding the Flat Torus.}
We must now find a parallelogram~$Q$ such that~$Q \cap \tilde{\R}$ is our desired rectangular dual
that also respects the embedding of $G$ on $\bT$.
For this, we use the meridian~$\bl M$ and horizon~$\re H$;
recall that each of them corresponds to a sequence of alternating faces and edges of $G$.
Define a curve~$\bl\cC_M$ that mirrors~$\bl M$ in~$\tilde{R}$ as follows.
Let~$f = \set{a,b,c}$ be the face in which~$\bl M$ and~$\re H$ cross each other, say, in patch~$P$. 
We start $\bl\cC_M$ at the point~$p_1$ where rectangles~$R_P(a), R_P(b)$ and~$R_P(c)$ meet.
Traversing along~$\bl M$, we next intersect an edge~$e$ to reach a face~$f'$.
Accordingly, we continue curve~$\bl \cC_M$ along the line segment that represents~$e$ in~$\R_P$.
If this brings us to the boundary of~$\R_P$, 
then find the patch~$P'$ that shares this boundary with~$P$, 
and translate~$\R_{P'}$ so that it abuts the side of~$\R_P$ that we have reached.
This could repeat once more (edge~$e$ could belong to a feedback closed walk~$C_i$ twice, but no more than that), 
but eventually we reach the end of~$e$'s segment and hence a point that represents face~$f'$.
We continue this process until for the first time we reach a point~$p_2$ that also represents~$f$
and have hence traversed the entirety of~$\bl M$.
This finishes~$\bl\cC_M$ and we obtain $\re\cC_H$ analogously.
Finally, we get two more curves~$\re\cC_H'$ and~$\bl\cC_M'$ by starting the process at~$p_1$ with~$\re H$ (to reach point~$p_4$) 
and starting the process at~$p_4$ with~$\bl M$; 
these curves are translates of~$\re\cC_H$ and~$\bl\cC_M$ since~$\tilde{\R}$ is periodic, and in particular~$\bl\cC_M'$ ends at~$p_3$.
See \cref{fig:example:tiles}.

The four curves~${\bl\cC_M}, {\re\cC_H}, {\bl\cC_M'}, {\re\cC_H'}$ together form a closed curve (call it~$Q'$).
Curve~$Q'$ may well touch itself, e.g. if~$\bl M$ and~$\re H$ cross the same edge, but we can use an
arbitrarily close simple curve to define `interior' of~$Q'$.
This interior contains every vertex-rectangle exactly once, 
because it corresponds to the flat torus on which~$G$ was given in which every vertex appears exactly once.
Therefore~$Q' \cap \tilde{\R}$ would be a rectangular dual on a flat torus 
if we permitted bizarrely shaped polygons (composed of suitably repeating polygonal paths) in place of parallelograms.
To obtain an actual parallelogram~$Q$, we simply use the four corner points~$p_1, p_2, p_3, p_4$ of~$Q'$.
Then~$Q \cap \tilde{\R}$ is the desired rectangular dual~$\R$; see \cref{fig:example:final}.

\subparagraph{Run-Time Analysis.}
While our proof of existence occasionally went over to the (infinite) UCP,
all vital steps to find the rectangular dual can be done efficiently.
Finding~$C_1$ and~$C_2$ takes linear time, as does splitting the graphs, computing the dual graphs, and computing longest paths.
Since any edge is visited at most four times in total by either~$C_1$ or~$C_2$, the total length of these circuits is~$\Oh(\abs{ V(G) })$,
so splitting graph~$G$ at both~$C_1$ and~$C_2$ gives~$\Oh(\abs{V(G)})$ vertices in total, which hence also
bounds the total size of all patches together.
So we can compute~$\R_P$ for all patches~$P$ in linear time.
Computing curves~$\cC_M$ and~$\cC_H$ can then be done in~$\Oh(n)$ time, where~$n$ is the input-size
(recall that this counts not only~$\abs{ V(G) }$ but also how many crossings there are on~$\bl M$ and~$\re H$).   
This also means that we had to translate~$\Oh(n)$ patches, a result that becomes useful below.

With this, we know~$\R$ implicitly:
We know~$Q$ from the endpoints of the curves, and
we know what the rectangular dual of each patch 
looks like, but we do not know how often each patch needs to be copied to where.
To compute the full rectangular dual~$\R$, we first must analyze how many patches could be relevant.
Recall that~$\tilde{\R}$ was obtained by first filling the entire plane with~$(d_2(s_2){-}1) \times (d_1(s_1){-}1)$-rectangles;
we call these the \emph{tiles}.    
Let~$\cT$ be the set of tiles that intersect~$Q$;
to obtain our rectangular dual we first compute~$\cT$
and then fill each tile with the rectangular dual of the corresponding patch.

To bound~$\abs{\cT}$, we split it into three sets.   
Let~$\cT_1$ be all those tiles that are intersected by the boundary of~$Q$.   
Let~$\cT_2$ be all those tiles that have~$L_\infty$-distance at most 2 from a tile in~$\cT_1$,
i.e., the tile lies within a~$5\times 5$-grid of tiles that is centered at a tile in~$\cT_1$.
Let~$\cT_3$ be all remaining tiles that intersect~$Q$.
We first bound~$\abs{\cT_1}$.
Recall that we translated~$\Oh(n)$ patches while computing~$\cC_M$, therefore~$\cC_M$ intersects~$\Oh(n)$ tiles.
It follows that the tiles containing~$p_1$ and~$p_2$ have~$L_\infty$-distance in~$\Oh(n)$; 
hence~$\overline{p_1p_2}$ can also only intersect~$\Oh(n)$ tiles.
So each side of~$Q$ intersects~$\Oh(n)$ tiles, or~$\abs{\cT_1}\in \Oh(n)$.   
This immediately implies~$\abs{\cT_2}\in \Oh(n)$ since each tile in~$\cT_1$ adds~$\Oh(1)$ tiles to~$\cT_2$.   
To bound~$\abs{\cT_3}$, consider a vertex~$v$ for which~$\R(v)$ lies fully within~$Q$.   
Observe that~$\R(v)$ can intersect at most a~$3\times 3$-grid of tiles, 
because an intersection with the horizontal [vertical] boundary of a tile means that~$\re C_2$ [$\bl C_1$] went through~$v$; 
this can happen at most twice in each direction.   
Therefore, each vertex whose rectangle lies fully inside~$Q$ gives rise to at most 9 tiles in~$\cT_3$.
Vice versa for any tile~$T$ in~$\cT_3$, there must be some vertex~$v$ for which~$\R(v)$ intersects~$T$.
Therefore, all of~$\R(v)$ lies within tiles that have~$L_\infty$-distance 2 from~$T$.   
None of these tiles is in~$\cT_1$ by definition of~$\cT_2$, so all of~$\R(v)$ lies within~$Q$.   
So the tiles intersected by vertices that lie fully in~$Q$ cover all tiles of~$\cT_3$,
and since each such vertex has only one rectangle in~$\R$ we have~$\abs{\cT_3}\in \Oh(\abs{ V(G) })$.

Putting it all together, we have~$\Oh(n)$ tiles that intersect~$Q$,
and computing the rectangular dual means translating $\Oh(n)$ patches and intersecting them with~$Q$,
which takes~$\Oh(\abs{ V(G) })$ time per patch.   
So the total time to compute~$\R$ is~$\Oh(n\abs{ V(G) })$.    
\end{proof}

We note here that the time-bound in \cref{clm:torus:slanted} is usually a vast overestimation; 
normally one would expect far fewer than $\Theta(n)$ tiles to intersect $Q$,
and there would be either few tiles or most of the corresponding patches would not have $\Theta(\abs{V(G)})$ rectangles.   
But we cannot rule out the possibility that a single patch (with $\Theta(\abs{V(G)})$ vertices) 
is intersected $\Theta(n)$ times by a side of $Q$; 
this could happen for example if both sides of $Q$ are nearly vertical and $Q$ hence becomes a very long and skinny sliver.

\subsection{Full Example of Slanted Toroidal Rectangular Dual}
\label{sec:example}

In this section, we illustrate the construction steps of a toroidal rectangular dual~$\R$
for a PTT graph~$G$ embedded on~$\bT$ with a slanted REL~$\cL$ as described in the proof of~\cref{clm:torus:slanted}. 
The steps are as follows:
\begin{itemize}
  \item For both $\blG$ and $\reG$, compute a left-first cycle $C_\ell$ and a right-first cycle $C_r$;
  uses these to compute feedback circuits $\bl C_1$ and $\re C_2$. See \cref{fig:example:feedback}.
  \item Use $\bl C_1$ and $\re C_2$ to dissect the cover graph $\tilde G$ into patches. See \cref{fig:example:patches}.
  \item Use $C_i$ to split $L_i(G)$ into the cylindrical graphs $L_i(G)\para C_i$, compute their duals~$L_i^\star(G)\para C_i$,
  and obtain the consistent numberings~$d_i(\cdot)$. See \cref{fig:example:duals}.
  \item Use the $d_i$ to compute rectangular duals of the patches. See \cref{fig:example:patchduals}.
  \item Tile the UCP into tiles of size $d_2(t_2) \times d_1(t_1)$ and (virtually)
  fill it with the rectangular duals of the patches to obtain the infinite rectangular dual~$\tilde{\R}$.
  By following the merdian $\bl M$ and the horizon $\bl H$ of $\bT$ in the UCP, 
  obtain curves $\cC_M$ and $\cC_H$ in~$\tilde{\R}$. See \cref{fig:example:tiles}.
  \item Based on the endpoints of $\cC_M$ and $\cC_H$, obtain the slanted flat torus~$Q$ filled with the toroidal rectangular dual of $G$.
  See \cref{fig:example:final}.
\end{itemize}

\begin{figure}[ht]
  \captionsetup[subfigure]{justification=centering}
  \centering
  \begin{subfigure}[t]{0.32 \linewidth}
  	\centering
	\includegraphics[page=1]{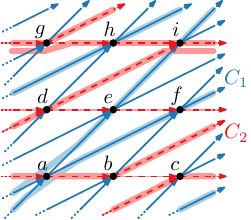}
	\caption{$G$ with $({\bl L_1}, {\re L_2})$ and\\ feedback circuits $\bl C_1$ and $\re C_2$}  
  \end{subfigure}
  \begin{subfigure}[t]{0.32	\linewidth}
  	\centering
	\includegraphics[page=3]{noOrbital2.pdf}
	\caption{$\blG$ with $\gr C_\ell$, $\og C_r$\\ and the resulting $\bl C_1$}  
  \end{subfigure}
  \begin{subfigure}[t]{0.32 \linewidth}
  	\centering
	\includegraphics[page=2]{noOrbital2.pdf}
	\caption{$\reG$ with $\cy C_\ell$, $\pu C_r$\\ and the resulting $\re C_2$}  
  \end{subfigure}
  \caption{Feedback circuits $C_i$ for $i\in [2]$.}
  \label{fig:example:feedback}
\end{figure}

\begin{figure}[ht]
  \centering
  \includegraphics[page=4,trim=35 36.5 70 67,clip]{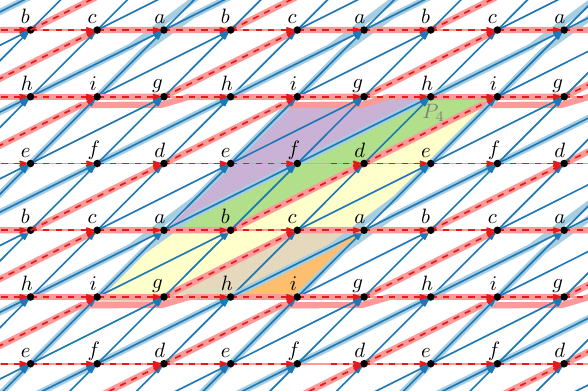}
  \caption{In the cover graph~$\tilde G$ of $G$ we can see the (here five) patches created by the feedback circuits $\bl C_1$ and $\re C_2$.}
  \label{fig:example:patches}
\end{figure}

\begin{figure}[ht]
  \centering
  \begin{subfigure}[t]{1 \linewidth}
  	\centering
	\includegraphics[page=5]{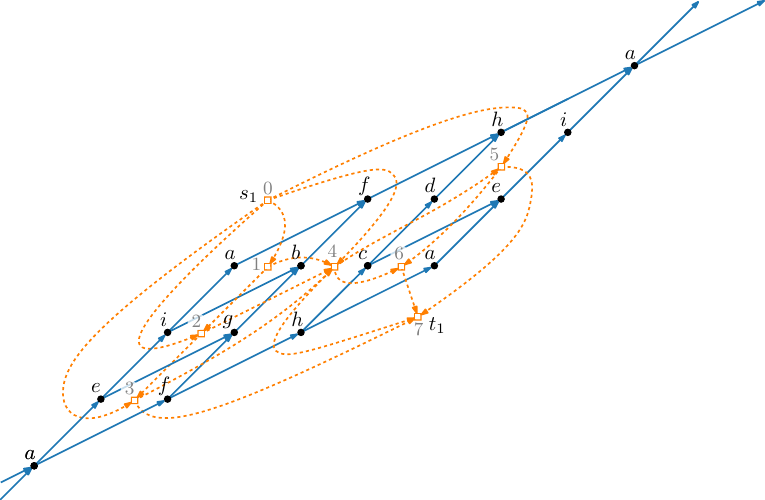}
	\subcaption{$\bl L_1 \para C_1$ and $\og L_1^\star(G)$}
  \end{subfigure}
  \\[3ex]
  \begin{subfigure}[t]{1 \linewidth}
  	\centering
	\includegraphics[page=6]{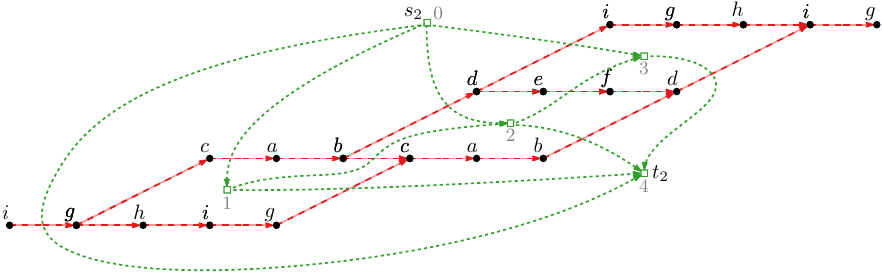}
	\subcaption{$\re L_2 \para C_2$ and $\gr L_2^\star(G)$}
  \end{subfigure}
  \caption{Compute the consistent numberings~$d_i(\cdot)$ in~$L_i^\star(G)\para C_i$ for $i\in [2]$.}
  \label{fig:example:duals}
\end{figure}

\begin{figure}[ht]
  \centering
  \includegraphics[page=7]{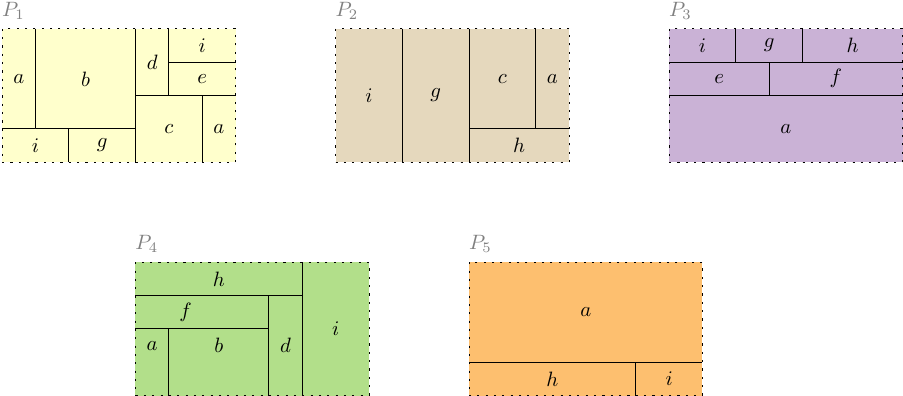}
  \caption{Rectangular duals of the five patches.}
  \label{fig:example:patchduals}
\end{figure}

\begin{figure}[ht]
  \centering
  \includegraphics[width=\linewidth,page=8]{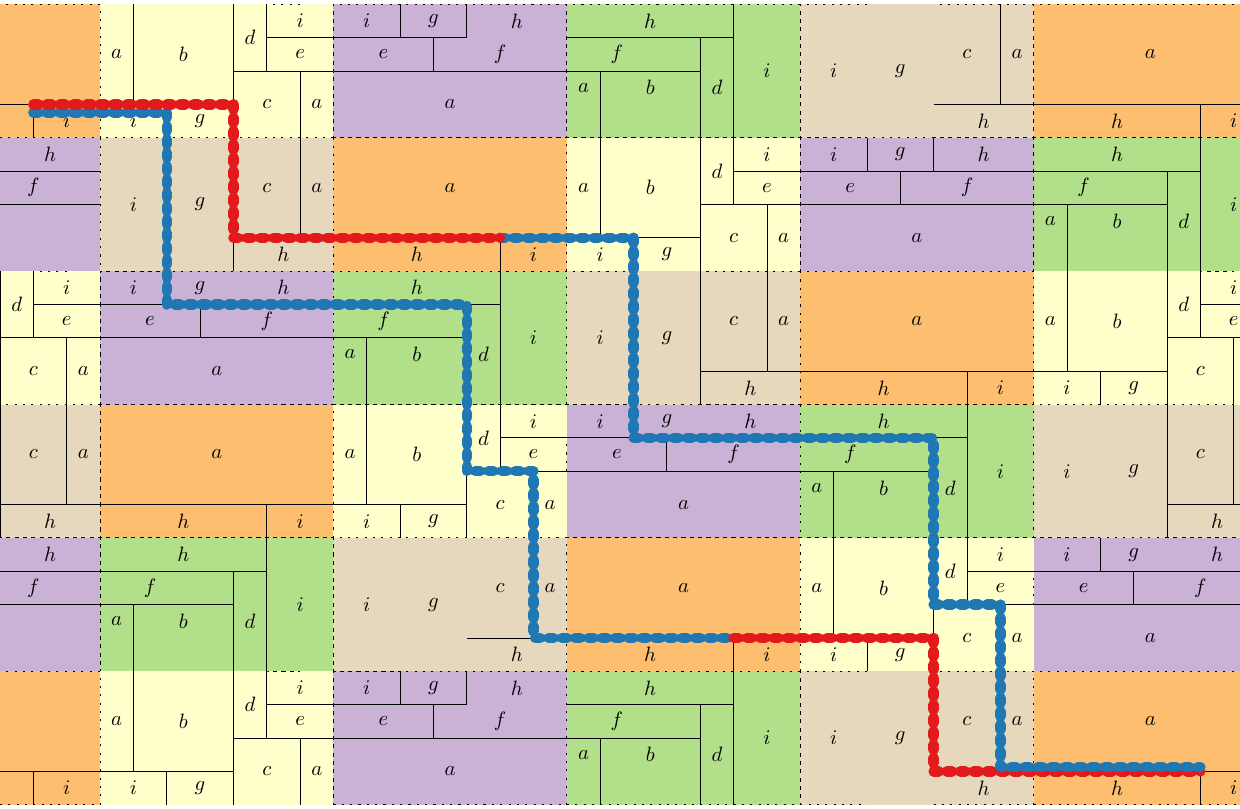}
  \caption{Putting the rectangular duals of the patches together in the tiling of the plane and obtaining curves $\cC_M$ and $\cC_H$.}
  \label{fig:example:tiles}
\end{figure}

\begin{figure}[ht]
  \centering
  \includegraphics[width=\linewidth,page=9]{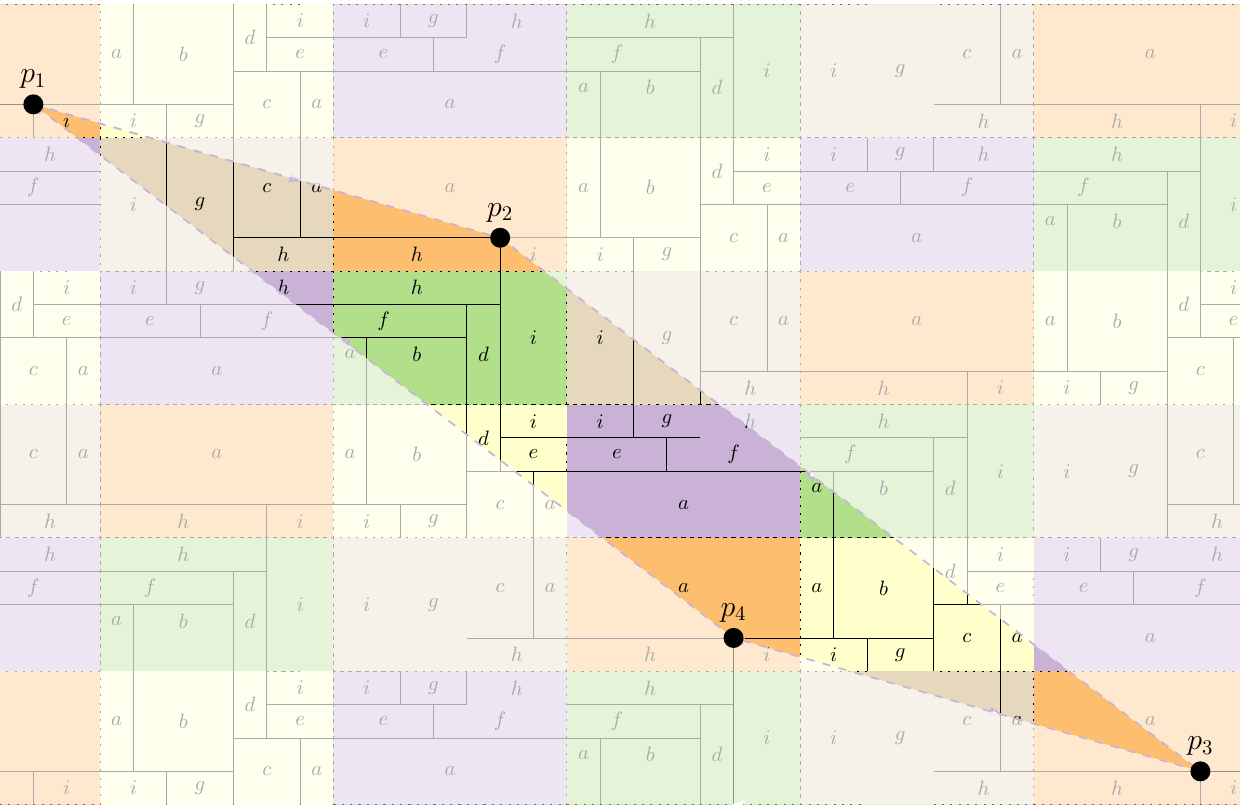}
  \caption{The final rectangular dual.}
  \label{fig:example:final}
\end{figure}

\FloatBarrier

\section{Rectangular Dual on the Cylinder} % ------------------------
\label{sec:cylinder}

Recall that a graph $G$ is a \emph{properly triangulated cylindrical (PTC) graph} 
if $G$ is cylindrical, internally triangulated, all loops, parallel edges, and separating triangles are non-contractible,
the outer faces have no non-consecutive occurrences of a vertex and no chord,
and the following \emph{degree-condition} holds: If $\set{u,v}$ is a parallel edge and $v$ is incident
to a loop, then $\deg(u) \geq 4$, and either $u$ is on an outer face, or $\deg(u) \geq 6$, or $\deg(u) = 5$
and at least one neighbour of $u$ is not incident to a loop.
In this section, we prove that every PTC graph admits a cylindrical rectangular dual.
We first show how to find a cylindrical REL and 
then how to construct the rectangular dual via reduction to the toroidal case.

\subsection{Cylindrical RELs: Construction via Canonical Orderings} % - - - - - - - - - - - 
\label{sec:cylindricalREL}
The construction of a cylindrical REL is heavily based on the idea of a~\emph{$(3,1)$-canonical order}, 
defined by Biedl and Derka~\cite{BD16}.
This is a partition of the vertices of a planar 4-connected triangulation 
into an ordered sequence of vertex sets that satisfies some conditions.    
This partition can be used to obtain a REL of a PTP graph.
Since we have additional conditions, we cannot use the approach from Biedl and Derka~\cite{BD16}
as a black box and instead re-do and expand it.
We need a helper-lemma, whose proof closely follows Biedl and Derka~\cite[Lemma~1]{BD16}

\begin{lemma} \label{clm:find}
  Let~$G$ be a PTC graph with outer faces~$f_s$ and~$f_t$ and with~$V(f_s) \neq V(G)$.  
  Then there exists a set~$V' \subset V(f_t)$ such that
  \begin{itemize}
	\item~$V'$ contains no vertex of~$f_s$,
	\item~$G \setminus V'$ is a PTC graph, and
	\item~$V'$ is either
	\begin{itemize}
		\item a single vertex~$z$ with~$\deg(z) \geq 4$ (a \emph{singleton}), 
		\item it induces a path~$\set{z_1, \dots, z_k}$ with~$\deg(z_i) = 3$ for~$i \in [k]$ (a \emph{fan}),
		\item $V' = V(f_t) = \set{z_1, \dots, z_k}$ with each $z_i$, $i \in [k]$, incident to a vertex $x$ that spans a loop (a \emph{loop neighborhood}), or
		\item $V' = V(f_t) = \set{z_1, z_2}$ and $\deg(z_1) = \deg(z_2) = 4$ and the two neighbours of $z_1, z_2$ span parallel edges (an \emph{enclosing parallel pair}).
	\end{itemize} 
  \end{itemize}
\end{lemma}
\begin{proof}
For ease of description, from now on we view~$G$ as a planar graph~$\hat{G}$;
one outer face~$f_t$ of~$G$ becomes the (unique) outer face of~$\hat{G}$
while the other outer face~$f_s$ of~$G$ becomes some inner face.
The condition that required that all loops/parallel edges/separating triangles
are non-contractible in~$G$ becomes in graph~$\hat{G}$ that they have~$f_s$ inside.
Enumerate the outer face vertices of~$\hat{G}$ as~$c_1, c_2,\dots, c_\ell$ in clockwise order.

We first consider the two special cases possible when $G$ is non-simple.

\medskip\indent\textbf{Case 1 -- loop neighborhood:}
Suppose that all vertices $c_1, c_2,\dots, c_\ell$ are adjacent to a vertex~$x$ that spans a loop; see \cref{fig:cylinder:canonical:loop}.
Since $f_s$ lies inside this loop, we know that $V(f_t) \cap V(f_s) = \emptyset$.
Set $V' = \set{c_1, c_2,\dots, c_\ell}$, a loop neighborhood. 
It is straightforward to check that $G \setminus V'$ is a PTC graph,
since $x$ becomes the only vertex on $f_t$.
Thus there cannot be a chord at $x$, the degree-condition holds for $x$ (and still for all other vertices),
and all loops/parallel edges/separating triangles have~$f_s$~inside.

\medskip\indent\textbf{Case 2 -- enclosing parallel pair:}
We first consider the special case that $f_t$ is bounded by a pair of parallel edges incident to $c_1$ and $c_2$
whose two neighbors $x$ and $y$ span a pair of parallel edges; see \cref{fig:cylinder:canonical:pair}.
(Note that neither $c_1$ nor $c_2$ can have a loop as this would be a chord of $f_t$.) 
Since $f_s$ lies insides this inner pair of parallel edges, we know that $V(f_t) \cap V(f_s) = \emptyset$.
Set $V' = \set{c_1, c_j}$, an enclosing parallel pair.

To see that $G \setminus V'$ is a PTC graph,
we have to show that neither $x$ nor $y$ can be incident to a loop (which would be a chord).
For the sake of contradiction, suppose that $x$ was incident to a loop, which necessarily 
contains $f_s$ and so must lie inside the pair of edges of $x$ and $y$. 
By the degree-condition (and since $y$ is then neither on $f_s$ nor on $f_t$) means $\deg(y) \geq 5$.
Vertex $y$ then cannot be incident to a loop since this loop could not separate faces $f_s$ and $f_t$.
Neither of the edges from $y$ to $c_1, c_2$ can be parallel edges (by $\deg(c_i) = 4$ for $i \in [2]$),
and so $y$ then must have a neighbour $w \neq x, y, c_1, c_2$.
But all faces incident to $y$ in the graph induced by $\set{x, y, c_1, c_2}$ are triangles
(possibly degenerate ones consisting of a pair of parallel edges and a loop) and do not
contain $f_s$, so this means that the face containing $w$ becomes a separating triangle, a contradiction.

\begin{figure}[t]
  \centering
  \begin{minipage}[t]{0.4\linewidth}
    \centering
    \includegraphics[page=10]{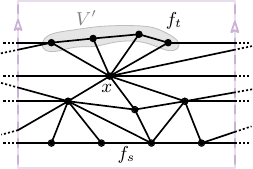}
	\subcaption{Case 1: $V'$ is a loop neighborhood.}
	\label{fig:cylinder:canonical:loop}
  \end{minipage}
  \hfil
  \begin{minipage}[t]{0.4\linewidth}
    \centering
  	\includegraphics[page=9]{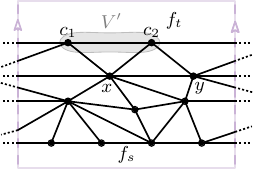}
	\subcaption{Case 2: $V'$ is an enclosing parallel pair.}
	\label{fig:cylinder:canonical:pair}
  \end{minipage}
  \caption{For \cref{clm:find}, first two cases of finding vertex set~$V'$ and obtaining a cylindrical REL.}
  \label{fig:cylinder:canoni}
\end{figure}

\medskip
For the next two cases we first need some definitions.
Assume that there is an inner vertex $x$ with two edges $\set{x,c_i}, \set{x,c_j}$ to vertices on the outer face.
Such a vertex splits the graph~$\hat{G}$ into two parts bounded by~$\croc{c_i, x, c_j}$
and one or the other path from~$c_i$ to~$c_j$ on the outer face.  
The \emph{outside [inside] of $\croc{c_i,x,c_j}$} is the part that does [does not] contain~$f_s$. 
We call $\croc{c_i,x,c_j}$ a \emph{2-leg} if its inside contains at least one other vertex of~$\hat{G}$;
vertex $x$ is then called a \emph{2-leg center}; see \cref{fig:cylinder:canonical:second}.
Note that $c_i = c_j$ is possible and allowed exactly when $c_i$ and $x$ span a pair of parallel edges. 
For a 2-leg-center~$x$ there may be multiple 2-legs if~$x$ has three or more neighbours on the outer face.
The \emph{maximal 2-leg} of~$x$ is the 2-leg $\croc{c_i, x, c_j}$ for which the outside is as small as possible.

\medskip\indent\textbf{Case 3 -- there exists a 2-leg:}
Then we can find the vertex set~$V'$ almost exactly as in Biedl and Derka~\cite{BD16}.   
Briefly, let~$x$ be a 2-leg center, and let~$\croc{c_i, x, c_j}$ be its maximal 2-leg; see \cref{fig:cylinder:canonical:leg}.
We assume that Case 1 does not apply, so if $x$ spans a loop than not all vertices on $f_t$ are adjacent to $x$.
Up to renaming of~$c_1, \dots, c_\ell$ and splitting\footnotemark $c_i$ if $c_i = c_j$,
we may assume that~$i < j$ and the inside of the 2-leg is bounded by~$x, c_i, c_{i+1}, \dots, c_{j-1}, c_j, x$.
Note that thus $i \leq j - 2$.
\footnotetext{Imagine splitting $c_i$ by cutting through $c_i$ from $f_t$ such that both the resulting copies remain adjacent to $x$.}
Call the inside~$H$ and let~$H^+$ be obtained from~$H$ by adding the edge~$\set{c_i, c_j}$ 
such that the outer face of~$H^+$ is the cycle~$c_i,c_{i+1},\dots,c_j,c_i$.
Observe that~$H$ has no loops, parallel edges, or separating triangles (they could not contain~$f_s$ inside),
and so all separating triangles of~$H^+$ contain edge~$\set{c_i, c_j}$.

Now apply the corresponding lemma from Biedl and Derka~\cite[Lemma 1]{BD16} to~$H^+$ 
(with outer face edge~$(c_i, c_j)$) to find a set~$V'$ that is either a singleton or a fan, 
contains only outer vertices (but not~$c_i, c_j$), 
and such that~${H^+} \setminus V'$ has no chord
and all separating triangles use edge~$(c_i, c_j)$.
One easily verifies that the same set~$V'$ works for our lemma.
\footnote{The lemma in Biedl and Derka~\cite{BD16} demands a graph
that has no separating triangles at all,
but as one easily verifies, the proof goes through verbatim as long as all
separating triangles must contain the pre-specified edge on the outer face
that is not allowed to be in~$V'$.}

\begin{figure}[t]
  \centering
  \begin{minipage}[t]{\linewidth}
    \centering
  	\begin{minipage}[b]{0.4\linewidth}
        \centering
        \includegraphics[page=1]{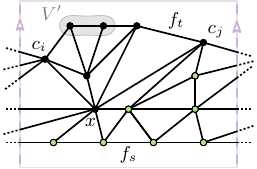}
    \end{minipage}
	\hfil
  	\begin{minipage}[b]{0.4\linewidth}
        \centering
        \includegraphics[page=2]{cylinder}
    \end{minipage}
	\subcaption{Case 3:~$\croc{c_i, x, c_j}$ forms a maximal 2-leg.
	Green vertices do not belong to~$H^+$.}
	\label{fig:cylinder:canonical:leg}
  \end{minipage}

  \medskip

  \begin{minipage}[t]{0.4\linewidth}
   \centering
   \includegraphics[page=3]{cylinder}
	\subcaption{Case 4:~$z$ avoids a separating triangle (shaded orange).}
	\label{fig:cylinder:canonical:otherwise}
  \end{minipage}
  \hfil
  \begin{minipage}[t]{0.4\linewidth}
   \centering
  	\includegraphics[page=5]{cylinder}
	\subcaption{The cylindrical REL. Gray regions indicate chosen sets~$V'$.}
	\label{fig:cylinder:canonical:full}
  \end{minipage}
  \caption{For \cref{clm:find}, next two cases of finding vertex set~$V'$ and obtaining a cylindrical REL.}
  \label{fig:cylinder:canonical}
\end{figure}

\medskip\indent\textbf{Case 4 -- there is no 2-leg:} 
We distinguish four subcases to pick $z$.\\
(i) If $f_t$ is a loop, let $z$ be the sole vertex on $f_t$.
Note that $\deg(z) \geq 4$ and $z \not\in V(f_s)$ as otherwise the loop would be a chord. \\
(ii) Suppose $f_t$ is bounded by of parallel edges on $c_1$ and $c_2$.
We assume that Case 2 does not apply.
If either $c_1$ or $c_2$ is part of a separating triangle~$T$ with edges not on $f_t$, then
let $z$ be the other one and observe that it is separated from~$V(f_s)$ by~$T$.
Otherwise, let $z$ be either of the two that is not in~$V(f_s)$. \\
(iii) Otherwise and if~$G$ contains no separating triangles, 
then let~$z$ be an arbitrary vertex of~$V(f_t) \setminus V(f_s)$, which exists by assumption.\\
(iv) Otherwise, let~$T$ be the separating triangle of~$G$ that maximizes the number of vertices inside.
Then any other separating triangle~$T'$ must be on or inside~$T$, for otherwise it either contains~$T$ inside
(contradicting the choice of~$T$), or it does not have~$f_s$ inside (contradicting non-contractability).
Since~$T$ is separating, not all vertices of~$f_t$ belong to~$T$; 
let~$z$ be any vertex of~$V(f_t) \setminus T$ and observe that it is separated from~$V(f_s)$ by~$T$;
see \cref{fig:cylinder:canonical:otherwise}.

We have~$\deg(z) \neq 2$, since otherwise its neighbors would form a chord.
We have~$\deg(z) \neq 3$, since otherwise its neighbors would form a 2-leg 
(at least one of its neighbours, say $x$, cannot be on~$f_t$, else $\set{x, z}$ would be a chord).
Therefore,~$\deg(z) \geq 4$ and~$V' = \set{ z }$ is a singleton vertex.
Assume for contradiction that~$\hat{G}' = \hat{G} \setminus z$ had a chord~$\set{x,y}$ on its outer face~$f_t'$.
Since~$\hat{G}$ had no chord on~$f_t$, at least one of the endpoints, say~$x$,
must not be on~$f_t$, so it is a neighbour of~$z$.
If~$y$ is on~$f_t$, then this would have made $\set{z, x, y}$ a 2-leg, contradicting the case. 
If~$y$ is not on~$f_t$ and $x \neq y$, then it is also a neighbour of~$z$,
and~$\set{z, x, y}$ is a separating triangle, contradicting the choice of~$z$.
If~$y$ is not on~$f_t$ and $x = y$, then~$\set{x,y}$ is a loop (because $G$ is triangulated)
and not a chord (and Case 1 applies).
The remaining properties are easily checked as before.
\end{proof}

\begin{lemma} \label{clm:cylindricalREL}
  For any PTC graph $G$, we can in linear time find a cylindrical REL
  for which each {\re red} edge lies on a {\re red} cycle and the {\bl blue} graph is~acyclic.
\end{lemma}
\begin{proof}
We proceed by induction on~$\abs{V(G)}$. Let~$f_s$ and~$f_t$ be the outer faces.
If~$V(G) = V(f_s)$, then~$G$ is a single cycle that bounds
both~$f_s$ and~$f_t$ (otherwise there would be a chord). Make the edges of
this cycle {\re red}, and orient them clockwise (by which we mean that~$f_s$ is
to the right of the directed cycle); this gives the desired REL.

If~$V(G) \neq V(f_s)$, apply~\cref{clm:find} to find~$V'$ and let~$G' = G \setminus V'$.
Since~$V'$ contains vertices of~$f_t$ and is disjoint from~$f_s$, 
one outer face of~$G'$ is again~$f_s$; let~$f_t'$ be the other outer face.
Since~$G'$ is a smaller PTC graph,  
we can find by induction a cylindrical REL for~$G'$ that satisfies the condition
and for which~$f_t'$ is oriented~clockwise and {\re red}. 

\begin{figure}[t]
  \centering
  \begin{minipage}[t]{0.32 \linewidth}
	\centering
    \includegraphics[page=1]{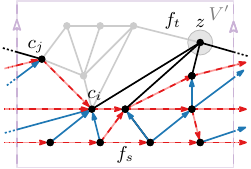}
  \end{minipage}
  \hfill
  \begin{minipage}[t]{0.32 \linewidth}
	\centering
     \includegraphics[page=2]{cylinderREL}
  \end{minipage}
  \hfill
  \begin{minipage}[t]{0.32 \linewidth}
	\centering
     \includegraphics[page=3]{cylinderREL}
  \end{minipage}
  \caption{Processing first a singleton $\set z$ and then a fan $\set{z_1, z_2}$ to construct a cylindrical REL.}
  \label{fig:cylinder:REL}
\end{figure}

First, suppose $V' = V(f_t)$.
Then either $V' = \set{z}$ and $f_t$ is a loop $e$ at $z$,
or $V' = \set{z_1, z_2}$ and an enclosing parallel pair,
or $V' = \set{z_1, \ldots, z_k}$ is a loop neighborhood. 
Either way, orient the edge(s) on $f_t$ clockwise and color them {\re red}.
Direct all other edges incident to~$V'$ towards~$V'$ and color them {\bl blue}.
Note that since~$G$ is internally triangulated, all vertices on $f_t'$ have neighbors in~$V'$
and vice versa.

Next, suppose $V' \neq V(f_t)$.
Enumerate~$f_t'$ as~$c_1, \dots, c_\ell$ in clockwise order
in such a way that edge~$c_1$ also belongs to~$f_t$.   
Let~$c_i, c_j$ be neighbors of~$V'$ on~$f_t'$ with~$i,j$ minimal and maximal, respectively.   
Since~$G$ is internally triangulated, all vertices in~$c_i, \dots, c_j$ have neighbors in~$V'$.  
If~$V'$ is a singleton, its unique vertex~$z_1$ is adjacent to all of~$c_i, \dots, c_j$.
If~$V'$ is a fan (enumerated as~$\langle z_1,\dots,z_k \rangle$ along the path, starting
with the neighbor of~$c_i$) then~$j = i{+}2$ and~$c_{i+1}$ is adjacent to all of~$V'$.   
Either way, we have a path~$\langle c_i,z_1,\dots,z_k,c_j \rangle$,
which we direct from~$c_i$ to~$c_j$ and color it {\re red}.  
(In particular, all vertices of~$V'$ receive incoming and outgoing {\re red} edges, and
the outer face~$f_t$ becomes a clockwise directed {\re red} cycle.) 
Direct all other edges incident to~$V'$ towards~$V'$ and color them {\bl blue}.
This gives an incoming {\bl blue} edge to every vertex of~$V'$, while all vertices in~$\set{c_{i+1}, \dots, c_{j-1}}$
(i.e.\ on~$f_t'$ but not on~$f_t$) receive an outgoing {\bl blue} edge towards~$V'$.
One verifies that this is a cylindrical REL.

We must show that every {\re red} edge belongs to a (directed) cycle and there are no {\bl blue} (directed) cycles.   
By induction this held for~$G'$.   
Adding the blue edges towards~$V'$ cannot create blue cycles since all vertices of~$V'$ become sinks in~$G$.
Every added red edge lies on~$f_t$, hence on a red cycle. 

The linear-time algorithm by Biedl and Derka~\cite{BD16} to compute a~$(3,1)$-canonical order 
can be straightforwardly adapted to compute the cylindrical REL in linear time.
Its amortized time to find set~$V'$ is~$\Oh(\abs{V'})$, 
and so the REL can be computed in linear time.
\end{proof}

\subsection{From Cylindrical REL to Cylindrical Rectangular Dual} % - - - - - - - - - - -
\label{sec:cylindricalRD} 
Now that we can compute a cylindrical REL of a PTC graph, 
we can construct a rectangular dual. 
Note that we indeed need a PTC graph.

\begin{lemma} \label{clm:cylinder:necessary}
  If a cylindrical graph $G$ has a rectangular dual, then~$G$ is a PTC~graph.
\end{lemma}
\begin{proof}
Having a cylindrical rectangular dual $\R$ implies 
that $G$ must be internally triangulated and has only non-contractible separating triangles, parallel edges, and loops.
For each outer face $f$ of~$G$, the rectangles $\R(v)$ for all vertices on $f$
must attach at an unidentified side of the flat cylinder in the order
in which they appear on $f$.   
Assuming that this side is horizontal, 
then for any two non-consecutive vertices $v,w$ on $f$, 
the two rectangles $\R(v), \R(w)$ cannot share an x-coordinate.   
Therefore, there is no edge $\set{v, w}$, i.e.\ $f$ has no chord. 
By the same argument we cannot have two occurrences of a vertex
$v$ on $f$, because its rectangle $\R(v)$ occupies a contiguous
part of the side.

If a rectangle $\R(v)$ loops around the cylinder,
then on either side is either an outer face or a rectangle $\R(u)$
that shares two horizontal segments with $\R(v)$.
Since $\R(u)$ also has at least one rectangle to the left and to the right (or touches itself),
and on the opposite side of $\R(u)$ either an outer face, 
a double contact with another looping rectangle $\R(w)$,
or at least one other rectangle, the degree of $v$ is at least four, six, or five, respectively.
Hence $G$ is a PTC graph.
\end{proof}

Note that the degree-constraint for PTC graph was not require in a PTT graph,
since in a toroidal rectangular dual parallel edges need not have the same color; see~\cref{fig:cylinderVStorus}.

\begin{figure}[ht]
  \captionsetup[subfigure]{justification=centering}
  \centering
  \begin{subfigure}[t]{0.33 \linewidth}
  	\centering
	\includegraphics[page=1]{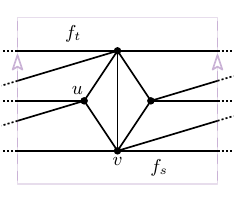}
	\caption{Cylindrical graph that does not admit a REL;
	the degree of $u$ is not high enough.}  
	\label{fig:cylinderVStorus:cylinder}
  \end{subfigure}
  $\quad$
  \begin{subfigure}[t]{0.33	\linewidth}
  	\centering
	\includegraphics[page=2]{cylinderNonSimple.pdf}
	\caption{Extension into a PTT graph with REL.}  
	\label{fig:cylinderVStorus:torus}
  \end{subfigure}
  $\quad$
  \begin{subfigure}[t]{0.25 \linewidth}
  	\centering
	\includegraphics[page=3]{cylinderNonSimple.pdf}
	\caption{Rectangular dual of the PTT graph.}  
	\label{fig:cylinderVStorus:recDual}
  \end{subfigure}
  \caption{Some non-simple cylindrical graphs cannot be realized by cylindrical rectangular duals
  but can appear as subgraphs of PTT graphs.}
  \label{fig:cylinderVStorus}
\end{figure}

\begin{theorem} \label{clm:cylinder}
  For a PTC graph~$G$ a cylindrical rectangular dual~$\R$ of~$G$ 
  that lies on a rectangular flat cylinder can be computed in linear time.
\end{theorem}
\begin{proof}
Let $\cL$ be a cylindrical REL of~$G$ computed with \cref{clm:cylindricalREL}.
We obtain a rectangular dual of~$G$ via reduction to toroidal graphs,
so we extend~$\cL$ to an orbital toroidal REL~$\hat{\cL}$.

Let $f_s$ and $f_t$ be the two outer faces.
There must be sources and sinks in the acyclic graph~$\blG$, 
and those must be on $f_s$ and $f_t$ since the REL property holds at inner vertices.   
Due to the REL property at outer vertices
(recall that we require that exactly one group of {\bl blue} edges is omitted) 
and since $f_s$ and $f_t$ are each bounded by a {\re red} cycle, 
either {\em all} vertices of $f_s$ are sources or {\em all} vertices of $f_s$ are sinks. 
Up to symmetry, assume that all vertices of~$f_s$ are sources and thus all vertices of $f_t$ are sinks of $\blG$.

For~$x \in \set{s,t}$, if~$f_x$ bounded by a loop, let~$v_x$ be that vertex.
Otherwise, we add new vertices~$v_t,v_s$, add for each vertex~$v$ of~$f_t$ a {\bl blue} edge~$\bl (v, v_t)$,
add two {\bl blue} edges~$\bl (v_t,v_s)$ as well as {\re red} loops at each of~$v_s$ and~$v_t$ (as non-contractible curves), and
add for each vertex~$w$ of~$f_s$ add a {\bl blue} edge~$\bl (v_s, w)$.
Finally, add a second {\bl blue} edge from one vertex~$v$ of~$f_t$ to~$v_t$
such that the two edges~$\bl (v, v_t)$ form an undirected cycle with homotopy class~$(1, 0)$;
similarly, add a second {\bl blue} edge from~$v_s$ to one vertex~$w$ of~$f_t$.
One verifies that the result is a PTT graph~$\hat{G}$; see also \cref{fig:cylinder:conversion}.
The assigned colors and directions complete~$\cL$ into a toroidal REL~$\hat{\cL}$ of~$\hat{G}$, 
since the REL property holds at~$v_s,v_t$ and all outer vertices of~$G$ due to the added edges.
REL~$\hat{\cL}$ is realizable since {\re red} edges lie on {\re red} cycles
and any {\bl blue} edge of~$G$ lies on a path of~$G$ from~$f_s$ to~$f_t$ (by the REL property) 
that can be completed into a blue cycle in~$\hat{G}$ via~$v_s$ and~$v_t$.
Define~$\re H$ to be the {\re red} loop at~$v_t$.   
Any {\re red} cycle~$C$ of~$\hat{\cL}$ is either~$\re H$, 
or it is completely disjoint from~$\re H$ since~$v_t$ has no other incident {\re red} edges.   
Therefore, all {\re red} cycles are parallel to~$\re H$, 
and with~$\re H$ being parallel to the horizon, we get that~$\hat{\cL}$ is {\re 2-orbital}.   

\begin{figure}[t]
  \begin{minipage}[t]{0.32 \linewidth}
	\centering
  	\includegraphics[page=6]{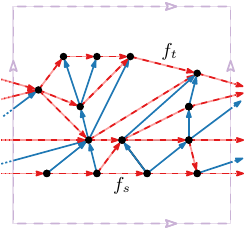}
	\subcaption{From cylindrical REL\ldots}
  \end{minipage}
  \hfill
  \begin{minipage}[t]{0.32 \linewidth}
	\centering
  	\includegraphics[page=7]{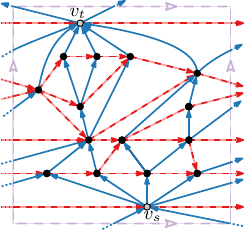}
	\subcaption{\ldots to orbital toroidal REL\ldots}
  \end{minipage}
  \hfill
  \begin{minipage}[t]{0.27 \linewidth}
	\centering
  \includegraphics[page=8]{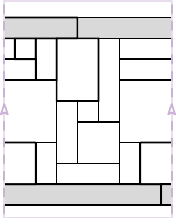}
	\subcaption{\ldots to rectangular dual.}
  \label{fig:cylinder:recdual}
  \end{minipage}
  \caption{To construct a rectangular dual of a PTC graph with a cylindrical REL,
  we extend it to a PTT graph and toroidal REL 
  so we can use the algorithm for toroidal rectangular duals.}
  \label{fig:cylinder:conversion}
\end{figure}

On the flat cylinder we can ``twist'' $G$ around the cylinder without changing its embedding,
yet on the flat torus we now have to fix an embedding.
In fact, we claim that~$\hat{\cL}$ is also {\bl 1-orbital} for a suitable choice of meridian~$\bl M$. 
Consider the left-first and right-first cycles~$C_\ell$ and~$C_r$
in~$\re L_1(\hat{G})$ obtained by starting at~$v_t$.  
These cycles diverge at~$v_t$ (using the two copies of~$(v_t, v_s)$), 
and again diverge at~$v_s$, since~$v_s$ has at least two outgoing {\bl blue} edges.
Since~$C_\ell$ and~$C_r$ cannot cross left-to-right, $\hati(C_r, C_\ell) \geq 2$.   
Using \cref{clm:combined}, we can find a {\bl blue} cycle that crosses both~$C_\ell$ and~$C_r$ algebraically; declare~$\bl M$ to be this cycle.
Since~$C_\ell$ and~$C_r$ can only be crossed left-to-right and right-to-left, respectively, we have~$m(C_\ell) < 0 < m(C_r)$.
Since~$C_\ell$ and~$C_r$ cross~$\re H$ at~$v_t$ right-to-left, we have~$h(C_\ell), h(C_r) > 0$.
So by \cref{clm:orbitalCharacterization} $\hat{\cL}$ is {\bl 1-orbital} and thus orbital. 

By \cref{clm:torus:relToRecDual}, we can use the orbital toroidal REL~$\hat{\cL}$ 
to obtain a rectangular dual~$\hat{\R}$ of~$\hat{G}$ on a rectangular flat torus~$Q$.
After translation, if needed, the top of~$\hat{\R}(v_t)$ coincides with the top of~$Q$, 
meaning that the two edges~$(v_t,v_s)$ are realized exactly across the horizontal side of~$Q$.
If we now delete $\R(v_t)$ and $\R(v_s)$ (unless we used existing framing vertices)
and interpret~$Q$ as a flat cylinder rather than a flat torus, 
then we exactly obtain a cylindrical rectangular dual of $G$; see~\cref{fig:cylinder:recdual}.
\end{proof}

\section{Concluding Remarks} % ------------------------
\label{sec:conclusion}
In this paper, we studied rectangular duals on the torus and on the cylinder.   
For toroidal graphs with a given REL, 
we characterized when this REL can correspond to a rectangular dual; 
if also given an embedding on the flat torus then we can test 
whether the corresponding flat torus can be made a rectangle.   
For cylindrical graphs, we need not be given a REL: 
we can characterize for any cylindrical graph whether it has a rectangular dual on the rectangular flat cylinder.   
The conditions can be tested, and if they hold, the corresponding rectangular duals can be found in linear time.

Our main open problem concerns toroidal graphs without REL and/or embedding on the flat torus.   
Does every properly-triangulated toroidal graph have a realizable REL?
Bonichon and L{\'{e}}v{\^{e}}que~\cite{BL19} show that it has a so-called `balanced 4-orientation'
from which a REL can be derived, but is this always realizable?
Secondly, given a REL~$\cL$, can we choose an embedding on $\bT$ that makes~$\cL$ orbital?
By exhaustive case analysis we can argue that this is not always possible for the graph in 
\cref{fig:otherrecduals}, not even if we are allowed to change the REL.
What are necessary and sufficient conditions, and can they be tested in linear time?

% ---- Bibliography ----
\pdfbookmark[1]{References}{References}
\bibliographystyle{plainurl}
\bibliography{sources}

\begin{thebibliography}{10}

\bibitem{AE0KPSU15}
Jawaherul Alam, David Eppstein, Michael Kaufmann, Stephen~G. Kobourov, Sergey
  Pupyrev, Andr{\'{e}} Schulz, and Torsten Ueckerdt.
\newblock Contact graphs of circular arcs.
\newblock In Frank Dehne, J{\"{o}}rg{-}R{\"{u}}diger Sack, and Ulrike Stege,
  editors, {\em {WADS'15}}, volume 9214 of {\em {LNCS}}, pages 1--13. Springer,
  2015.
\newblock \href {https://doi.org/10.1007/978-3-319-21840-3_1}
  {\path{doi:10.1007/978-3-319-21840-3_1}}.

\bibitem{ADDFPR20}
Patrizio Angelini, Giordano {Da Lozzo}, Giuseppe {Di Battista}, Fabrizio Frati,
  Maurizio Patrignani, and Ignaz Rutter.
\newblock Beyond level planarity: Cyclic, torus, and simultaneous level
  planarity.
\newblock {\em Theoretical Computer Science}, 804:161--170, 2020.
\newblock \href {https://doi.org/10.1016/j.tcs.2019.11.024}
  {\path{doi:10.1016/j.tcs.2019.11.024}}.

\bibitem{BFL24}
Olivier Bernardi, {\'E}ric Fusy, and Shizhe Liang.
\newblock Grand-schnyder woods.
\newblock {\em Annals of Combinatorics}, 2024.
\newblock \href {https://doi.org/10.1007/s00026-024-00729-8}
  {\path{doi:10.1007/s00026-024-00729-8}}.

\bibitem{BS87}
Jayaram Bhasker and Sartaj Sahni.
\newblock A linear time algorithm to check for the existence of a rectangular
  dual of a planar triangulated graph.
\newblock {\em Networks}, 17(3):307--317, 1987.
\newblock \href {https://doi.org/10.1002/net.3230170306}
  {\path{doi:10.1002/net.3230170306}}.

\bibitem{Bie22}
Therese Biedl.
\newblock Visibility representations of toroidal and {K}lein-bottle graphs.
\newblock In Patrizio Angelini and Reinhard von Hanxleden, editors, {\em Graph
  Drawing and Network Visualization (GD)}, volume 13764, pages 404--417.
  Springer, 2022.
\newblock \href {https://doi.org/10.1007/978-3-031-22203-0_29}
  {\path{doi:10.1007/978-3-031-22203-0_29}}.

\bibitem{BD16}
{Therese} {Biedl} and {Martin} {Derka}.
\newblock The (3,1)-ordering for 4-connected planar triangulations.
\newblock {\em Journal of Graph Algorithms and Applications}, 20(2):347--362,
  2016.
\newblock \href {https://doi.org/10.7155/jgaa.00396}
  {\path{doi:10.7155/jgaa.00396}}.

\bibitem{BL19}
Nicolas Bonichon and Benjamin L{\'{e}}v{\^{e}}que.
\newblock A bijection for essentially 4-connected toroidal triangulations.
\newblock {\em Electronic Journal of Combinatorics}, 26(1):1, 2019.
\newblock \href {https://doi.org/10.37236/7897} {\path{doi:10.37236/7897}}.

\bibitem{BGPV08}
Adam~L. Buchsbaum, Emden~R. Gansner, Cecilia~Magdalena Procopiuc, and Suresh
  Venkatasubramanian.
\newblock Rectangular layouts and contact graphs.
\newblock {\em {ACM} Transactions on Algorithms}, 4(1):8:1--8:28, 2008.
\newblock \href {https://doi.org/10.1145/1328911.1328919}
  {\path{doi:10.1145/1328911.1328919}}.

\bibitem{CDF13}
Luca Castelli~Aleardi, Olivier Devillers, and {\'E}ric Fusy.
\newblock Canonical ordering for triangulations on the cylinder, with
  applications to periodic straight-line drawings.
\newblock In Walter Didimo and Maurizio Patrignani, editors, {\em Graph Drawing
  (GD'12)}, volume 7704 of {\em {LNCS}}, pages 376--387. Springer, 2013.
\newblock \href {https://doi.org/10.1007/978-3-642-36763-2_34}
  {\path{doi:10.1007/978-3-642-36763-2_34}}.

\bibitem{CEGL11}
Erin~W Chambers, David Eppstein, Michael~T Goodrich, and Maarten L{\"o}ffler.
\newblock Drawing graphs in the plane with a prescribed outer face and
  polynomial area.
\newblock In {\em Graph Drawing (GD'10)}, volume 6502 of {\em {LNCS}}, pages
  129--140. Springer, 2011.
\newblock \href {https://doi.org/10.1007/978-3-642-18469-7_12}
  {\path{doi:10.1007/978-3-642-18469-7_12}}.

\bibitem{CELP21}
Erin~Wolf Chambers, Jeff Erickson, Patrick Lin, and Salman Parsa.
\newblock How to morph graphs on the torus.
\newblock In D{\'{a}}niel Marx, editor, {\em {ACM-SIAM} Symposium on Discrete
  Algorithms (SODA'21)}, pages 2759--2778. {SIAM}, 2021.
\newblock \href {https://doi.org/10.1137/1.9781611976465.164}
  {\path{doi:10.1137/1.9781611976465.164}}.

\bibitem{CEX15}
Hsien{-}Chih Chang, Jeff Erickson, and Chao Xu.
\newblock Detecting weakly simple polygons.
\newblock In Piotr Indyk, editor, {\em {ACM-SIAM} Symposium on Discrete
  Algorithms (SODA'15)}, pages 1655--1670. {SIAM}, 2015.
\newblock \href {https://doi.org/10.1137/1.9781611973730.110}
  {\path{doi:10.1137/1.9781611973730.110}}.

\bibitem{CFKKRW22}
Steven Chaplick, Stefan Felsner, Philipp Kindermann, Jonathan Klawitter, Ignaz
  Rutter, and Alexander Wolff.
\newblock Simple algorithms for partial and simultaneous rectangular duals with
  given contact orientations.
\newblock {\em Theoretical Computer Science}, 919:66--74, 2022.
\newblock \href {https://doi.org/10.1016/j.tcs.2022.03.031}
  {\path{doi:10.1016/j.tcs.2022.03.031}}.

\bibitem{CKKRW22}
Steven Chaplick, Philipp Kindermann, Jonathan Klawitter, Ignaz Rutter, and
  Alexander Wolff.
\newblock Morphing rectangular duals.
\newblock In Patrizio Angelini and Reinhard von Hanxleden, editors, {\em Graph
  Drawing and Network Visualization (GD'22)}, volume 13764 of {\em {LNCS}},
  pages 389--403. Springer, 2022.
\newblock \href {https://doi.org/10.1007/978-3-031-22203-0_28}
  {\path{doi:10.1007/978-3-031-22203-0_28}}.

\bibitem{CDMB20}
Kun-Ting Chen, Tim Dwyer, Kim Marriott, and Benjamin Bach.
\newblock {DoughNets:} visualising networks using torus wrapping.
\newblock In Regina Bernhaupt, Florian~`Floyd' Mueller, David Verweij, Josh
  Andres, Joanna McGrenere, Andy Cockburn, Ignacio Avellino, Alix Goguey,
  Pernille Bj{\o}n, Shengdong Zhao, Briane~Paul Samson, and Rafal Kocielnik,
  editors, {\em Conference on Human Factors in Computing Systems ({CHI'20})},
  pages 1--11. {ACM}, 2020.
\newblock \href {https://doi.org/10.1145/3313831.3376180}
  {\path{doi:10.1145/3313831.3376180}}.

\bibitem{dLDEJ17}
Giordano Da~Lozzo, William~E. Devanny, David Eppstein, and Timothy Johnson.
\newblock Square-contact representations of partial 2-trees and triconnected
  simply-nested graphs.
\newblock In Yoshio Okamoto and Takeshi Tokuyama, editors, {\em International
  Symposium on Algorithms and Computation (ISAAC'17)}, volume~92 of {\em
  {LIPIcs}}, pages 24:1--24:14, 2017.
\newblock \href {https://doi.org/10.4230/LIPICS.ISAAC.2017.24}
  {\path{doi:10.4230/LIPICS.ISAAC.2017.24}}.

\bibitem{dFdMR94}
Hubert de~Fraysseix, Patrice~Ossona de~Mendez, and Pierre Rosenstiehl.
\newblock On triangle contact graphs.
\newblock {\em Combinatorics, Probability and Computing}, 3(2):233--246, 1994.
\newblock \href {https://doi.org/10.1017/S0963548300001139}
  {\path{doi:10.1017/S0963548300001139}}.

\bibitem{Dean00}
A.~Dean.
\newblock A layout algorithm for bar-visibility graphs on the {M}{\"{o}}bius
  band.
\newblock In Joe Marks, editor, {\em Graph Drawing (GD'00)}, volume 1984 of
  {\em {LNCS}}, pages 350--359. Springer, 2000.
\newblock \href {https://doi.org/10.1007/3-540-44541-2_33}
  {\path{doi:10.1007/3-540-44541-2_33}}.

\bibitem{DGK11}
Christian~A. Duncan, Michael~T. Goodrich, and Stephen~G. Kobourov.
\newblock Planar drawings of higher-genus graphs.
\newblock {\em Journal of Graph Algorithms and Applications}, 15(1):7--32,
  2011.
\newblock \href {https://doi.org/10.7155/jgaa.00215}
  {\path{doi:10.7155/jgaa.00215}}.

\bibitem{Epp17}
David Eppstein.
\newblock Triangle-free penny graphs: Degeneracy, choosability, and edge count.
\newblock In Fabrizio Frati and Kwan{-}Liu Ma, editors, {\em Graph Drawing and
  Network Visualization (GD'17)}, volume 10692 of {\em {LNCS}}, pages 506--513.
  Springer, 2017.
\newblock \href {https://doi.org/10.1007/978-3-319-73915-1_39}
  {\path{doi:10.1007/978-3-319-73915-1_39}}.

\bibitem{EMSV12}
David Eppstein, Elena Mumford, Bettina Speckmann, and Kevin Verbeek.
\newblock Area-universal and constrained rectangular layouts.
\newblock {\em {SIAM} Journal on Computing}, 41(3):537--564, 2012.
\newblock \href {https://doi.org/10.1137/110834032}
  {\path{doi:10.1137/110834032}}.

\bibitem{EL23}
Jeff Erickson and Patrick Lin.
\newblock Planar and toroidal morphs made easier.
\newblock {\em Journal of Graph Algorithms and Applications}, 27(2):95--118,
  2023.
\newblock \href {https://doi.org/10.7155/jgaa.00616}
  {\path{doi:10.7155/jgaa.00616}}.

\bibitem{Fus09}
{\'{E}}ric Fusy.
\newblock Transversal structures on triangulations: A combinatorial study and
  straight-line drawings.
\newblock {\em Discrete Mathematics}, 309(7):1870--1894, 2009.
\newblock \href {https://doi.org/10.1016/j.disc.2007.12.093}
  {\path{doi:10.1016/j.disc.2007.12.093}}.

\bibitem{GS69}
K.~Ruben Gabriel and Robert~R. Sokal.
\newblock A new statistical approach to geographic variation analysis.
\newblock {\em Systematic Biology}, 18(3):259--278, 1969.
\newblock \href {https://doi.org/10.2307/2412323} {\path{doi:10.2307/2412323}}.

\bibitem{GL14}
Daniel Gon{\c{c}}alves and Benjamin L{\'{e}}v{\^{e}}que.
\newblock Toroidal maps: Schnyder woods, orthogonal surfaces and straight-line
  representations.
\newblock {\em Discrete \& Computational Geometry}, 51(1):67--131, 2014.
\newblock \href {https://doi.org/10.1007/s00454-013-9552-7}
  {\path{doi:10.1007/s00454-013-9552-7}}.

\bibitem{He93}
Xin He.
\newblock On finding the rectangular duals of planar triangular graphs.
\newblock {\em {SIAM} Journal on Computing}, 22(6):1218--1226, 1993.
\newblock \href {https://doi.org/10.1137/0222072} {\path{doi:10.1137/0222072}}.

\bibitem{Hli01}
Petr Hlinen{\'{y}}.
\newblock Contact graphs of line segments are np-complete.
\newblock {\em Discrete Mathematics}, 235(1-3):95--106, 2001.
\newblock \href {https://doi.org/10.1016/S0012-365X(00)00263-6}
  {\path{doi:10.1016/S0012-365X(00)00263-6}}.

\bibitem{HK01}
Petr Hlinen{\'{y}} and Jan Kratochv{\'{\i}}l.
\newblock Representing graphs by disks and balls (a survey of
  recognition-complexity results).
\newblock {\em Discrete Mathematics}, 229(1-3):101--124, 2001.
\newblock \href {https://doi.org/10.1016/S0012-365X(00)00204-1}
  {\path{doi:10.1016/S0012-365X(00)00204-1}}.

\bibitem{KH97}
Goos Kant and Xin He.
\newblock Regular edge labeling of 4-connected plane graphs and its
  applications in graph drawing problems.
\newblock {\em Theoretical Computer Science}, 172(1):175--193, 1997.
\newblock \href {https://doi.org/10.1016/S0304-3975(95)00257-X}
  {\path{doi:10.1016/S0304-3975(95)00257-X}}.

\bibitem{KNU15}
Jonathan Klawitter, Martin Nöllenburg, and Torsten Ueckerdt.
\newblock Combinatorial properties of triangle-free rectangle arrangements and
  the squarability problem.
\newblock In Emilio Di~Giacomo and Anna Lubiw, editors, {\em Graph Drawing and
  Network Visualization (GD'15)}, pages 231--244. Springer, 2015.
\newblock \href {https://doi.org/10.1007/978-3-319-27261-0_20}
  {\path{doi:10.1007/978-3-319-27261-0_20}}.

\bibitem{Koe36}
Paul Koebe.
\newblock {Kontaktprobleme der konformen Abbildung}.
\newblock {\em {Berichte {\"u}ber die Verhandlungen der S{\"a}chsischen
  Akademie der Wiss. zu Leipzig. Math.-Phys. Klasse 88}}, pages 141--164, 1936.

\bibitem{KK85}
Krzysztof Ko{\'z}mi{\'n}ski and Edwin Kinnen.
\newblock Rectangular duals of planar graphs.
\newblock {\em Networks}, 15(2):145--157, 1985.
\newblock \href {https://doi.org/10.1002/net.3230150202}
  {\path{doi:10.1002/net.3230150202}}.

\bibitem{KP96}
Jan Kratochv{\'i}l and Teresa Przytycka.
\newblock Grid intersection and box intersection graphs on surfaces.
\newblock In Franz~J. Brandenburg, editor, {\em Graph Drawing (GD'95)}, volume
  1027 of {\em {LNCS}}, pages 365--372. Springer, 1996.
\newblock \href {https://doi.org/10.1007/BFb0021820}
  {\path{doi:10.1007/BFb0021820}}.

\bibitem{KS24}
Vinod Kumar and Krishnendra Shekhawat.
\newblock On the characterization of rectangular duals.
\newblock {\em Notes on Number Theory and Discrete Mathematics},
  30(1):141--149, 2024.
\newblock \href {https://doi.org/10.7546/nntdm.2024.30.1.141-149}
  {\path{doi:10.7546/nntdm.2024.30.1.141-149}}.

\bibitem{LL90}
Yen{-}Tai Lai and Sany~M. Leinwand.
\newblock A theory of rectangular dual graphs.
\newblock {\em Algorithmica}, 5(4):467--483, 1990.
\newblock \href {https://doi.org/10.1007/BF01840399}
  {\path{doi:10.1007/BF01840399}}.

\bibitem{LPVV01}
Francis Lazarus, Michel Pocchiola, Gert Vegter, and Anne Verroust.
\newblock Computing a canonical polygonal schema of an orientable triangulated
  surface.
\newblock In Diane~L. Souvaine, editor, {\em Symposium on Computational
  Geometry (SoCG)}, pages 80--89. {ACM}, 2001.
\newblock \href {https://doi.org/10.1145/378583.378630}
  {\path{doi:10.1145/378583.378630}}.

\bibitem{LL84}
Sany~M. Leinwand and {Yen-Tai} Lai.
\newblock An algorithm for building rectangular floor-plans.
\newblock In {\em 21st Design Automation Conference}, pages 663--664, 1984.
\newblock \href {https://doi.org/10.1109/DAC.1984.1585874}
  {\path{doi:10.1109/DAC.1984.1585874}}.

\bibitem{MR98}
Bojan Mohar and Pierre Rosenstiehl.
\newblock Tessellation and visibility representations of maps on the torus.
\newblock {\em Discrete \& Computational Geometry}, 19:249--263, 1998.
\newblock \href {https://doi.org/10.1007/PL00009344}
  {\path{doi:10.1007/PL00009344}}.

\bibitem{NK16}
Sabrina Nusrat and Stephen~G. Kobourov.
\newblock The state of the art in cartograms.
\newblock {\em Computer Graphics Forum}, 35(3):619--642, 2016.
\newblock \href {https://doi.org/10.1111/cgf.12932}
  {\path{doi:10.1111/cgf.12932}}.

\bibitem{Ste73}
Philip Steadman.
\newblock Graph theoretic representation of architectural arrangement.
\newblock {\em Architectural Research and Teaching}, pages 161--172, 1973.

\bibitem{stillwell}
John Stillwell.
\newblock {\em Classical topology and combinatorial group theory}, volume~72 of
  {\em Graduate Texts in Mathematics}.
\newblock Springer-Verlag, 1980.

\bibitem{SHRMH18}
Shaheena Sultana, Iqbal Hossain, Saidur Rahman, Nazmun~Nessa Moon, and Tahsina
  Hashem.
\newblock On triangle cover contact graphs.
\newblock {\em Computational Geometry: Theory and Applications}, 69:31--38,
  2018.
\newblock \href {https://doi.org/10.1016/J.COMGEO.2017.11.001}
  {\path{doi:10.1016/J.COMGEO.2017.11.001}}.

\bibitem{TT91}
R.~Tamassia and I.~Tollis.
\newblock Representations of graphs on a cylinder.
\newblock {\em {SIAM} Journal of Discrete Mathematics}, 4(1):139--149, 1991.
\newblock \href {https://doi.org/10.1137/0404014} {\path{doi:10.1137/0404014}}.

\bibitem{Tar72}
Robert Tarjan.
\newblock Depth-first search and linear graph algorithms.
\newblock {\em {SIAM} Journal on Computing}, 1(2):146--160, 1972.
\newblock \href {https://doi.org/10.1137/0201010} {\path{doi:10.1137/0201010}}.

\bibitem{Tho86}
Carsten Thomassen.
\newblock Interval representations of planar graphs.
\newblock {\em Journal of Combinatorial Theory, Series B}, 40(1):9--20, 1986.
\newblock \href {https://doi.org/10.1016/0095-8956(86)90061-4}
  {\path{doi:10.1016/0095-8956(86)90061-4}}.

\bibitem{Ung53}
Peter Ungar.
\newblock On diagrams representing maps.
\newblock {\em Journal of the London Mathematical Society}, 28(3):336--342,
  1953.
\newblock \href {https://doi.org/10.1112/jlms/s1-28.3.336}
  {\path{doi:10.1112/jlms/s1-28.3.336}}.

\bibitem{YS95}
Gary K.~H. Yeap and Majid Sarrafzadeh.
\newblock Sliceable floorplanning by graph dualization.
\newblock {\em SIAM Journal on Discrete Mathematics}, 8(2):258--280, 1995.
\newblock \href {https://doi.org/10.1137/S0895480191266700}
  {\path{doi:10.1137/S0895480191266700}}.

\end{thebibliography}

\end{document}